[1994/12/01]
\documentclass[manuscript, a4paper, reqno]{aomart}

\chardef\bslash=`\\ 

\newtheorem[{}\it]{thm}{Theorem}[section]
\newtheorem{cor}[thm]{Corollary}
\newtheorem{lem}[thm]{Lemma}

\newtheorem{prop}[thm]{Proposition}

\theoremstyle{remark} 

\usepackage{dsfont}
\usepackage{amssymb}
\usepackage[misc]{ifsym}

\usepackage{graphicx}
\usepackage{moresize}

\usepackage{hyperref}
\usepackage[numeric]{amsrefs}

\usepackage[english]{babel}

\theoremstyle{definition}
\newtheorem{defn}{\textsc{Definition}}[section]
\newtheorem{rem}{Remark}[section]

\newtheorem*[{}\it]{notation}{Notation}

\newtheorem*[{}\it]{rest}{\textsc{Theorem}}
\newtheorem*[{}\it]{quest}{\textsc{Question}}
\newtheorem*[{}\it]{problemo}{\textsc{Problem}}
\newtheorem*[{}\it]{projone}{\texttt{Project 1 (Future Work)}}
\newtheorem*[{}\it]{projtwo}{\texttt{Project 2 (Current Work)}}
\newtheorem*[{}\it]{projthree}{\texttt{Project 3 (Current Work)}}
\newtheorem*[{}\it]{projfour}{\texttt{Project 4 (Future Work)}}
\newtheorem*[{}\it]{projfive}{\texttt{Project 5 (Future Work)}}
\newtheorem*[{}\it]{proofoflemma}{Proof of Lemma}

\usepackage[textwidth=6in,textheight=9.25in, paperwidth=7.5in,paperheight=11.25in, inner=2.5cm, outer=2.5cm]{geometry} 

\usepackage[usenames,dvipsnames]{xcolor}

\usepackage{txfonts}

\usepackage{multicol}

\usepackage{enumerate}

\title[]{On the Initial Boundary-Value Problem \\ in the Kinetic Theory of Hard Particles II: Non-uniqueness}
\author[]{Mark Wilkinson}\thanks{Department of Mathematics and Computer Science (MACS), Heriot-Watt University, and the Maxwell Institute for Mathematical Sciences, Edinburgh, Scotland (\Letter) \href{mailto:mark.wilkinson@hw.ac.uk}{mark.wilkinson@hw.ac.uk}}

\allowdisplaybreaks

\newcommand{\ov}{\overline}

\newcommand{\mres}{\mathbin{\vrule height 1.6ex depth 0pt width
0.13ex\vrule height 0.13ex depth 0pt width 1.3ex}}

\hypersetup{citecolor=red}

\usepackage{tikz}
\usetikzlibrary{patterns}
\newcommand{\boundellipse}[3]
{(#1) ellipse (#2 and #3)
}

\usepackage{mathrsfs}

\begin{document}

\maketitle

\begin{abstract}
\noindent We prove that to each initial datum in a set of positive measure in phase space, there exist uncountably-many associated weak solutions of Newton's equations of motion which govern the dynamics of two non-spherical sets with real-analytic boundaries subject to the conservation of linear momentum, angular momentum and kinetic energy. We prove this result by first exhibiting non-uniqueness of classical solution to a constrained Monge-Amp\`{e}re equation posed on Euclidean space, and then applying the deep existence theory of \textsc{Ballard} for hard particle dynamics. In the final section of the article, we discuss the relevance of this observation to the kinetic theory of hard particle systems.
\end{abstract}
%

\vspace{3mm}

\section{\label{sec:level1}Introduction}
In \textsc{Saint-Raymond and Wilkinson} \cite{lsrmw}, a rigorous study of the Boltzmann equation governing a gas of non-spherical particles was initiated. If the underlying gas particles are each congruent to some compact, connected, strictly-convex set $\mathsf{P}_{\ast}\subset\mathbb{R}^{3}$ of unit mass whose boundary $\partial\mathsf{P}_{\ast}$ admits the structure of a manifold of class $\mathscr{C}^{1}$, in the case of {\em linear scattering} the Boltzmann equation for the unknown 1-particle distribution function $f=f(z, t)$ reads as
\begin{equation}\label{beebop}
\frac{\partial f}{\partial t}+\{f, H\}=\mathcal{C}[f, f],
\end{equation}
where $z=(x, R, v, \omega)\in\mathcal{M}_{3}:=\mathbb{R}^{3}\times\mathrm{SO}(3)\times\mathbb{R}^{3}\times\mathbb{R}^{3}$ is the 1-particle phase vector, $\{\cdot, \cdot\}$ denotes the classical Poisson bracket, and $\mathcal{C}$ is the quadratic collision operator defined pointwise by
\begin{equation*}
(\mathcal{C}[f, f])(z, t):=\int_{\mathrm{SO}(3)}\int_{\mathbb{R}^{3}}\int_{\mathbb{R}^{3}}\int_{\mathbb{S}^{2}}b_{\beta}(v, \ov{v}, \omega, \ov{\omega})(f_{\beta}'\ov{f}_{\beta}'-f\ov{f})\,dS(n)\ov{v}d\ov{\omega}d\mu(\ov{R}),
\end{equation*}
with $b_{\beta}:\mathbb{R}^{12}\rightarrow\mathbb{R}$ denoting the scattering cross section defined for each spatial parameter $\beta=(R, \ov{R}, n)\in\mathrm{SO}(3)\times\mathrm{SO}(3)\times\mathbb{S}^{2}$, $S$ denoting the normalised measure on $\mathbb{S}^{2}$, and $\mu$ denoting the Haar measure on $\mathrm{SO}(3)$. The values of the distribution function are expressed as
\begin{equation*}
\begin{array}{c}
f_{\beta}':=f(x, R, v_{\beta}', \omega_{\beta}', t), \quad \ov{f}_{\beta}':=f(x, \ov{R}, \ov{v}_{\beta}', \ov{\omega}_{\beta}', t), \vspace{2mm}\\
f:=f(x, R, v, \omega, t), \quad \ov{f}:=f(x, \ov{R}, \ov{v}, \ov{\omega}, t),
\end{array}
\end{equation*}
and the values of the post-collisional velocities are expressed by means of a {\em scattering matrix} $\sigma_{\beta}\in \mathrm{O}(12)$ acting on pre-collisional velocities, namely 
\begin{equation*}
[v_{\beta}', \ov{v}_{\beta}', \omega_{\beta}', \ov{\omega}_{\beta}']:=\sigma_{\beta}[v, \ov{v}, \omega, \ov{\omega}].
\end{equation*} 
Moreover, $H:\mathcal{M}_{3}\rightarrow\mathbb{R}$ is the time-independent 1-particle Hamiltonian given by
\begin{equation*}
H(z):=\frac{1}{2}|v|^{2}+\frac{1}{2}RJR^{T}\omega\cdot \omega,
\end{equation*} 
and $J\in\mathrm{GL}(3)$ is the inertia tensor associated to the reference set $\mathsf{P}_{\ast}$, i.e.
\begin{equation*}
J:=\int_{\mathbb{R}^{3}}\left(I-y\otimes y\right)\mathds{1}_{\mathsf{P}_{\ast}}(y)\,dy.
\end{equation*}
Although a `natural' family of scattering matrices in \textsc{Saint-Raymond and Wilkinson} \cite{lsrmw} was studied, the question of  possible {\em non-uniqueness} of global-in-time physical weak solutions of the initial boundary-value problem (IBVP) on phase space $\mathscr{D}_{N}(\mathsf{P}_{\ast})\subseteq \mathscr{P}_{N}(\mathsf{P}_{\ast})\times\mathbb{R}^{6N}$ associated to Newton's equations of motion
\begin{equation}\label{nemo}
\left\{
\begin{array}{ll}
\displaystyle \frac{dx_{i}}{dt}=v_{i}, & \quad \displaystyle\frac{dR_{i}}{dt}=\Omega_{i}R_{i}, \vspace{2mm}\\
\displaystyle \frac{dv_{i}}{dt}=0, & \quad \displaystyle\frac{d\omega_{i}}{dt}=0,
\end{array}
\right. \tag{NEM}
\end{equation}
was not investigated in full. In the above, $i=1, ..., N$ and $\mathscr{P}_{N}(\mathsf{P}_{\ast})$ denotes the $N$-particle table consisting of all admissible translations and rotations of the reference particle $\mathsf{P}_{\ast}$ given by
\begin{equation*}
\mathscr{P}_{N}(\mathsf{P}_{\ast}):=\left\{
\{(x_{j}, R_{j})\}_{j=1}^{N}\in (\mathbb{R}^{3}\times\mathrm{SO}(3))^{N}\,:\,\begin{array}{c}
\mathrm{card}\,(R_{i}\mathsf{P}_{\ast}+x_{i})\cap(R_{j}\mathsf{P}_{\ast}+x_{j})\leq 1 \\
\mathrm{for}\hspace{2mm}i\neq j
\end{array}
\right\}.
\end{equation*}
Rather, in that article, the authors focussed on questions pertinent at the level of {\em collision invariants} for the Boltzmann equation \eqref{beebop}, not those pertinent at the level of Newton's equations. This work addresses the issue of non-uniqueness of the collision operator for compact, strictly-convex hard particles.

In \textsc{Wilkinson} \cite{mw111}, it was shown that for initial data in a certain codimension 1 subset of phase space $\mathscr{D}_{N}(\mathsf{P}_{\ast})$, one cannot even construct a local-in-time weak solution of this IBVP equipped with any so-called frictionless scattering on the spatial boundary of $N$-particle phase space. This is due to a kind of `loss of convexity' at certain points of $\partial\mathscr{P}_{N}(\mathsf{P}_{\ast})$ when $\mathsf{P}_{\ast}$ does not admit spherical symmetry. Of course, if one is content to study only {\em average} properties of weak solutions to \eqref{nemo} (e.g. as is only required in the construction of weak solutions to the associated BBGKY hierarchy, c.f. \textsc{Lanford} \cite{lanford}), then the results of identifies a phenomenon which is essentially `undetected' by integrals built with respect to the restriction measure $((\mathscr{L}_{3N}\otimes \mu_{N})\mres \mathscr{P}_{N}(\mathsf{P}_{\ast}))\otimes\mathscr{L}_{6N}$, where $\mathscr{L}_{K}$ denotes the Lebesgue measure on $\mathbb{R}^{K}$ and $\mu_{N}$ denotes the Haar measure on the product group $\mathrm{SO}(3)^{N}$. 

In this article, we prove a non-uniqueness result for global-in-time physical weak solutions of the IBVP associated to \eqref{nemo}. More precisely, we prove that to each initial datum in a set of positive measure in phase space there exist {\em uncountably-many} global-in-time weak solutions which conserve the total linear momentum, angular momentum and kinetic energy of the initial datum. We contend that this result {\em is} relevant from the point of view of the theory of the Boltzmann equation, as it is {\em a priori} unclear if each these scattering families all admit the same collision invariants. We discuss this matter in greater detail in section \ref{discuss} below.
\subsection{On the Strategy of the Paper}
In order to establish our main non-uniqueness result, we study families of flow operators on phase space $\mathscr{D}_{N}(\mathsf{P}_{\ast})$. Indeed, by first assuming some natural regularity, stability and measure-theoretic criteria on (presumably-existing and global-in-time) weak solutions of Newton's equations, we derive in turn a suite of necessary conditions they -- and the associated {\em scattering maps} to which they give rise -- satisfy. By employing the important general existence theory of \textsc{Ballard} \cite{ball}, we are in turn able to demonstrate that not only (i) more than 1 scattering map satisfies the derived necessary conditions, but (ii) that each of these scattering maps gives rise to a distinct global-in-time weak solution of Newton's equations on $\mathbb{R}$ subject to the same initial datum.

For notational simplicity, we study systems of gas particles modelled by compact, connected subsets of $\mathbb{R}^{2}$, however all our results admit a straightforward extension to the case of three-dimensional particles. We prove that to each (what we term in this work) {\em physical regular flow} $\{T_{t}\}_{t\in\mathbb{R}}$ there is a unique family of {\em canonical scattering maps} $\{\sigma_{\beta}\}_{\beta}$ defined on velocity space, each of whose members $\sigma_{\beta}:=\nabla S_{\beta}$ is a classical solution of either an elliptic or hyperbolic Monge-Amp\`{e}re equation of the type
\begin{equation*}
\mathrm{det}\,D^{2}S_{\beta}=\pm1 \quad \text{on}\hspace{2mm}\mathbb{R}^{6}\hspace{2mm}\text{for}\hspace{2mm}\beta\in\mathbb{T}^{3}.
\end{equation*}
We prove that this equation -- whose solutions are subject to additional algebraic constraints coming from the conservation laws -- admits at least 2 classical solutions for fixed $\beta$ when $\nabla S_{\beta}$ is an orientation-reversing map on $\mathbb{R}^{6}$. In addition, we prove the same equation admits uncountably-many classical solutions -- parameterised by elements of the real projective line $\mathbb{RP}^{1}$ -- for fixed $\beta$ when $\nabla S_{\beta}$ is an orientation-preserving map on $\mathbb{R}^{6}$. Finally, by employing the general existence theory of \textsc{Ballard} \cite{ball} for a system of 2 real-analytic sets, we construct uncountably-many distinct one-parameter groups of solution operators $\mathsf{T}_{\phi}:=\{T(t; \phi)\}_{t\in\mathbb{R}}$ on rigid set phase space, where $\phi\in C^{0}(\mathbb{T}^{2}, \mathbb{RP}^{1})$, with the property that for each initial datum $Z_{0}$ in a set of full measure, the trajectories on $\mathbb{R}$
\begin{equation*}
t\mapsto T(t; \psi)Z_{0}
\end{equation*}
are distinct weak solutions of \eqref{nemo} for every choice $\phi\in C(\mathbb{T}^{2}, \mathbb{RP}^{1})$. In other words, our class 
\begin{equation*}
\{\mathsf{T}_{\phi}\,:\,\phi\in C^{0}(\mathbb{T}^{2}, \mathbb{RP}^{1})\}
\end{equation*}
of solution operator families $\mathsf{T}_{\phi}$ is parameterised by a family of continuous line fields on the 2-torus. 
\subsection{Non-uniqueness of Solutions to Differential Equations of Classical Physics}
We remark that our main result (but certainly not the proof thereof) is reminiscent of those in the recent and growing body of work focussing on possible non-uniqueness of weak solutions to equations of classical physics. An archetype of this kind of result is that of \textsc{De Lellis and Sz\'{e}kelyhidi} \cite{delellis1, delellis2} on the incompressible Euler equations of hydrodynamics. We recall for the reader a basic definition.
\begin{defn}[Weak Solution of the Incompressible Euler Equations on $\mathbb{R}^{d}$]
Suppose $d=2, 3$. We say that a time-dependent solenoidal vector field $u\in L^{2}_{\mathrm{loc}}(\mathbb{R}^{d}\times \mathbb{R}, \mathbb{R}^{d})$ is a {\em weak solution} of the incompressible Euler equations on $\mathbb{R}^{d}$ ,
\begin{equation}\label{incompeuler}
\left\{
\begin{array}{l}
\partial_{t}u+\mathrm{div}_{x}(u\otimes u)+\nabla p=0, \vspace{2mm}\\
\nabla_{x}\cdot u =0,
\end{array}
\right.
\end{equation}
if and only if $u$ satisfies the equality
\begin{equation*}
\int_{-\infty}^{\infty}\int_{\mathbb{R}^{d}}\left(u\cdot\partial_{t}\varphi+u\otimes u:\nabla_{x}\varphi\right)\,dxdt=0
\end{equation*}
for every divergence-free test vector field $\varphi\in C^{\infty}_{c}(\mathbb{R}^{d}\times\mathbb{R}, \mathbb{R}^{d})$.
\end{defn}
Building upon and extending the original pioneering work of \textsc{Scheffer} \cite{scheffer} (and also the work of \textsc{Shnirelman} \cite{shnir}), the authors established an elementary proof of the following non-uniqueness result for weak solutions:
\begin{thm}\label{weirdflow}
There exist $u\in L^{\infty}(\mathbb{R}^{d}\times\mathbb{R}, \mathbb{R}^{d})$ and $p\in L^{\infty}(\mathbb{R}^{d}\times\mathbb{R}, \mathbb{R})$ such that $u$ is a non-zero weak solution of \eqref{incompeuler} with the property that $\mathrm{supp}\,u$ and $\mathrm{supp}\,p$ are compact in $\mathbb{R}^{d}\times\mathbb{R}$.
\end{thm}
Theorem \ref{weirdflow} is manifestly a non-uniqueness result, as the trivial vector field $u\equiv 0$ on $\mathbb{R}^{d}\times\mathbb{R}$ is also a weak solution of \eqref{incompeuler}. For us, this result admits the following interpretation: {\em the laws of classical physics (together with, perhaps, constitutive relations) are not guaranteed to determine uniqueness of their associated equations of motion}. From the mathematician's viewpoint, it may be that one must supplement PDE of the type \eqref{incompeuler} with more stringent {\em analytical} conditions or so-called {\em selection principles} in the hope of securing a statement on the uniqueness of solutions thereof. In this vein, one thinks of the admissibility criterion of \textsc{Lax} \cite{lax1971shock, lax1973hyperbolic} in the setting of hyperbolic conservation laws.

As we have claimed above in the introduction, in this article we establish the non-uniqueness of {\em physical weak solutions} to the governing equations of rigid body dynamics; see section \ref{mainres} below for our definition of physical weak solution of system \eqref{nemo} in the case $N=2$. We shall argue in the final section of this work it seems quite hopeless that one can find a selection criterion which singles out any one of these solutions in a mathematically- or physically-natural way. Instead, with the Boltzmann-Grad limit of system \eqref{nemo} as $N\rightarrow\infty$ in mind, we believe it to be more appropriate to show -- in a precise sense -- that the average dynamics of \eqref{nemo} is independent of one's choice of scattering. This remark articulates an open problem concerning the extension of the work of \textsc{Saint-Raymond and Wilkinson} \cite{lsrmw} to general families of (linear) scattering maps. This work will not be tackled here.
\subsection{Notation}\label{notation}
In what follows, we work only with two congruent particles in space $\mathbb{R}^{2}$ (i.e. $N=2$ in \eqref{nemo} above). The phase space $\mathscr{D}_{2}(\mathsf{P}_{\ast})$ is the fibre bundle 
\begin{equation*}
\mathscr{D}_{2}(\mathsf{P}_{\ast}):=\bigsqcup_{X\in\mathscr{P}_{2}(\mathsf{P}_{\ast})}A_{X},
\end{equation*}
where
\begin{equation*}
A_{X}=\left\{
\begin{array}{ll}
\displaystyle \mathbb{R}^{6} & \quad \text{if}\hspace{2mm}X\in\mathrm{int}\,\mathscr{P}_{2}(\mathsf{P}_{\ast}), \vspace{2mm}\\
\displaystyle \Sigma_{\beta(X)}^{-} & \quad \text{if}\hspace{2mm}X\in\partial\mathscr{P}_{2}(\mathsf{P}_{\ast}),
\end{array}
\right.
\end{equation*}
and $\Sigma_{\beta(X)}^{-}$ are half-spaces of velocity vectors in $\mathbb{R}^{6}$ derived in section \ref{scat} below. We subsequently denote by $M\in\mathbb{R}^{6\times 6}$ the mass-inertia matrix associated to 2 congruent sets given by 
\begin{equation}\label{mi}
M:=\mathrm{diag}(\sqrt{m}, \sqrt{m}, \sqrt{m}, \sqrt{m}, \sqrt{J}, \sqrt{J}),
\end{equation}
where
\begin{equation*}
m:=\int_{\mathsf{P}_{\ast}}\,dy \quad \text{and}\quad J:=\int_{\mathsf{P}_{\ast}}|y|^{2}\,dy.
\end{equation*}
We shall employ square brackets to denote the concatenation of scalars and vectors into a single vector, e.g. $X=[x, \ov{x}, \vartheta, \ov{\vartheta}]\in\mathbb{R}^{6}$ for $x, \ov{x}\in \mathbb{R}^{2}$ and $\vartheta, \ov{\vartheta}\in\mathbb{R}$. On the other hand, round brackets denote concatenation of scalars alone. We write $\Pi_{1}:\mathscr{D}_{2}(\mathsf{P}_{\ast})\rightarrow\mathscr{P}_{2}(\mathsf{P}_{\ast})$ to denote the spatial projector
\begin{equation}
\Pi_{1}Z:=X\quad \text{for}\hspace{2mm}Z=[X, V]\in\mathscr{D}_{2}(\mathsf{P}_{2}),
\end{equation}
and write $\Pi_{2}:\mathscr{D}_{2}(\mathsf{P}_{\ast})\rightarrow\mathbb{R}^{6}$ to denote the velocity projector
\begin{equation}
\Pi_{2}Z:=V\quad \text{for}\hspace{2mm}Z=[X, V]\in\mathscr{D}_{2}(\mathsf{P}_{2}).
\end{equation}
Finally, we write $\mathsf{P}(t):=R(\vartheta(t))\mathsf{P}_{\ast}+x(t)$ and $\ov{\mathsf{P}}(t):=R(\ov{\vartheta}(t))\mathsf{P}_{\ast}+\ov{x}(t)$ to denote the evolution of the sets congruent to $\mathsf{P}_{\ast}$ that are governed by the abstract phase map $t\mapsto X(t)$.
\subsection{Main Result}\label{mainres}
We hereby present the precise definition of physical weak solution of \eqref{nemo} to be employed in all the sequel. In all that follows, we study the motion of $N=2$ hard particles. By a suitable modification, one can establish analogous results for $N$-particle problems on $\mathscr{D}_{N}(\mathsf{P}_{\ast})$.
\begin{defn}[Global-in-time Physical Weak Solutions]\label{physweak}
Suppose $\mathsf{P}_{\ast}\subset\mathbb{R}^{2}$ is a compact, strictly-convex set whose boundary $\partial\mathsf{P}_{\ast}$ is a real-analytic manifold. We say that $X\in C^{0}(\mathbb{R}, \mathscr{P}_{2}(\mathsf{P}_{\ast}))$ is a {\bf global-in-time physical weak solution} of 
\begin{equation}\label{newty}
M\ddot{X}=0
\end{equation}
with initial state $Z_{0}:=[X_{0}, V_{0}]\in\mathscr{P}_{2}(\mathsf{P}_{\ast})$ if and only if $\dot{X}\in\mathrm{BV}_{\mathrm{loc}}(\mathbb{R}, \mathbb{R}^{6})$ and
\begin{equation*}
\int_{\mathbb{R}}X(t)\cdot\phi''(t)\,dt=\int_{\mathbb{R}}\phi(t)\,d\mu(t)
\end{equation*}
for all $\phi\in C^{\infty}_{c}(\mathbb{R}, \mathbb{R}^{6})$, where $\mu$ is a $\mathbb{R}^{6}$-valued Radon measure on $\mathbb{R}$ with the support property that
\begin{equation}\label{suppy}
\mathrm{supp}\,\mu\hspace{2mm}\left\{
\begin{array}{ll}
=\varnothing & \quad \text{if}\hspace{2mm}\mathsf{P}(t)\cap\ov{\mathsf{P}}(t)\neq \varnothing\hspace{2mm}\text{on any open non-empty subset of}\hspace{1mm}\mathbb{R}, \\
\text{is finite} & \quad \text{otherwise}.
\end{array}
\right.
\end{equation}
Moreover, $X$ and its distributional derivative $\dot{X}$ respect the conservation of linear momentum
\begin{equation*}
mv(t)+m\ov{v}(t)=mv_{0}+m\ov{v}_{0},
\end{equation*}
the conservation of angular momentum
\begin{equation*}
\begin{array}{c}
-m(a-x(t))^{\perp}\cdot v(t)+J\omega(t)-m(a-\ov{x}(t))^{\perp}\cdot \ov{v}(t)\\
=-m(a-x_{0})^{\perp}\cdot v_{0}+J\omega_{0}-m(a-\ov{x}_{0})^{\perp}\cdot \ov{v}_{0}
\end{array}
\end{equation*}
for any $a\in\mathbb{R}^{2}$, and the conservation of kinetic energy
\begin{equation*}
\begin{array}{c}
m|v(t)|^{2}+J\omega(t)^{2}+m|\ov{v}(t)|^{2}+J\ov{\omega}(t)^{2}=m|v_{0}|^{2}+J\omega_{0}^{2}+m|\ov{v}_{0}|^{2}+J\ov{\omega}_{0}^{2}
\end{array}
\end{equation*}
for all $t\in\mathbb{R}$, where $t\mapsto V(t):=[v(t), \ov{v}(t), \omega(t), \ov{\omega}(t)]$ denotes the unique lower semi-continuous representative of the equivalence class $\dot{X}\in\mathrm{BV}_{\mathrm{loc}}(\mathbb{R}, \mathbb{R}^{6})$. Finally, $X(0)=X_{0}$ and $V_{0}\in \dot{X}$.
\end{defn}
\begin{rem}
The support condition \eqref{suppy} in the definition of solution is included so as to rule out the possibility of {\em rolling solutions} of Newton's equations in this context. In addition, the assumption that the support of $\mu$ be finite is motivated by the result of \textsc{Ballard} (\cite{ball}, Proposition 19) that ``kinetic energy-conserving solutions admit only finitely-many collisions on compact time intervals''.
\end{rem}
We are now ready to state the main result of this article.
\begin{thm}\label{mainrez}
Suppose $\mathsf{P}_{\ast}\subset\mathbb{R}^{2}$ is a compact, strictly-convex set with real-analytic boundary curve $\partial\mathsf{P}_{\ast}$. There exists a set $\mathcal{A}\subset\mathscr{D}_{2}(\mathsf{P}_{\ast})$ of initial data of positive measure with the following property: for all $Z_{0}\in \mathcal{A}$, there exist uncountably-many distinct global-in-time weak solutions of \eqref{newty}.
\end{thm}
We claim that this result has significant implications for the Boltzmann equation for non-spherical particles. We discuss this in more detail in section \ref{discuss} below. We also claim (without proof) that Theorem \ref{mainrez} may extended in a straightforward manner to the case of compact, strictly-convex subsets of $\mathbb{R}^{3}$.
\subsection{Structure of the Paper}
In section \ref{odeder}, we introduce the notion of regular flow on phase space $\mathscr{D}_{2}(\mathsf{P}_{\ast})$. In section \ref{scat}, we define scattering maps that are associated to a given regular flow on $\mathscr{D}_{2}(\mathsf{P}_{\ast})$. In section \ref{maspy}, we consider the Monge-Amp\`{e}re scattering problem that arises in the study of physical regular flows. Finally, in section \ref{discuss}, we discuss the implications of Theorem \ref{mainrez} for the kinetic theory of gases.
\section{Regular Flows}\label{odeder}
To define the class of flows on which we shall focus in this work, we must introduce a few basic definitions. We begin with the following.
\begin{defn}
For $M\geq 1$, let $\mathrm{LSC}(\mathbb{R}, \mathbb{R}^{M})$ and $\mathrm{USC}(\mathbb{R}, \mathbb{R}^{M})$ denote the vector spaces of lower- and upper-semicontinuous\footnote{We say a vector-valued map is lower semi-continuous on $\mathbb{R}$ if and only if each of its component functions is lower semi-continuous on $\mathbb{R}$.}  maps on $\mathbb{R}$. We define the operators 
\begin{equation*}
L:\mathrm{USC}(\mathbb{R}, \mathbb{R}^{M})\rightarrow\mathrm{LSC}(\mathbb{R}, \mathbb{R}^{M})
\end{equation*} and 
\begin{equation*}
U:\mathrm{LSC}(\mathbb{R}, \mathbb{R}^{M})\rightarrow\mathrm{USC}(\mathbb{R}, \mathbb{R}^{M})
\end{equation*}
by
\begin{equation*}
(Lf)(t):=\lim_{s\rightarrow t-}f(s) \quad \text{and}\quad (Ug)(t):=\lim_{s\rightarrow t+}g(s)
\end{equation*}
for each $t\in\mathbb{R}$ and all $f\in \mathrm{USC}(\mathbb{R}, \mathbb{R}^{M})$ and $g\in \mathrm{LSC}(\mathbb{R}, \mathbb{R}^{M})$.
\end{defn}
We now specify classes of regular maps which we shall use to model the `physical' evolution of non-spherical particles.
\begin{defn}
For $M\geq 1$ and $k\geq 0$, $C_{-}^{k}(\mathbb{R}, \mathbb{R}^{M})$ denotes the vector space of $k$-times left-differentiable maps on $\mathbb{R}$, i.e. $f\in C^{k}_{-}(\mathbb{R}, \mathbb{R}^{M})$ if and only if
\begin{equation*}
\lim_{h\rightarrow 0-}\frac{f^{(i)}(t+h)-f^{(i)}(t)}{h}\quad \text{exists for every}\hspace{1mm}t\in\mathbb{R},
\end{equation*}
and each $i=0, ..., k-1$. Similarly, $C_{+}^{k}(\mathbb{R}, \mathbb{R}^{M})$ denotes the vector space of $k$-times right-differentiable maps on $\mathbb{R}$.
\end{defn}
In general, the maps with which we work will belong to $C^{k}_{\pm}(\mathbb{R}, \mathbb{R}^{M})\setminus C^{k}(\mathbb{R}, \mathbb{R}^{M})$ due to the `presence of collisions' in the dynamics of hard sets. For this reason, we establish the following definition.
\begin{defn}[Collision Times]
Suppose $X\in C^{0}(\mathbb{R}, \mathscr{P}_{2}(\mathsf{P}_{\ast}))$. The set of {\bf collision times} $\mathcal{T}(X)\subseteq\mathbb{R}$ for the map $X$ is given by
\begin{equation*}
\mathcal{T}(X):=\left\{
t\in\mathbb{R}\,:\,\mathsf{P}(t)\cap\ov{\mathsf{P}}(t)\neq \varnothing
\right\}.
\end{equation*}
If the map $X$ is uniquely determined by an initial datum $X_{0}\in\mathscr{P}_{2}(\mathsf{P}_{\ast})$, we denote $\mathcal{T}(X)$ simply by $\mathcal{T}(X_{0})$.
\end{defn}
We now able to define sets of pre- and post-collisional velocities for suitably-smooth maps on $\mathbb{R}$.
\begin{defn}[Pre- and Post-collisional Velocities]
Suppose $\{T_{t}\}_{t\in\mathbb{R}}$ is a 1-parameter group of operators $T_{t}:\mathscr{D}_{2}(\mathsf{P}_{\ast})\rightarrow\mathscr{D}_{2}(\mathsf{P}_{\ast})$ with identity $T_{0}:=\mathrm{id}_{\mathscr{D}_{2}(\mathsf{P}_{\ast})}$, endowed with the group operation $T_{s}\circ T_{t}:=T_{s+t}$ for all $s, t\in\mathbb{R}$. Suppose further that $\{T_{t}\}_{t\in\mathbb{R}}$ has the property that the maps $t\mapsto (\Pi_{2}\circ T_{t})(Z_{0})$ belong to $\mathrm{LSC}(\mathbb{R}, \mathbb{R}^{6})$ for every $Z_{0}\in\mathscr{D}_{2}(\mathsf{P}_{\ast})$. The set $\mathcal{V}^{-}(X)\subseteq\mathbb{R}^{6}$ of velocities which are {\bf pre-collisional} with respect to a collision configuration $X\in\partial\mathscr{P}_{2}(\mathsf{P}_{\ast})$ is given by
\begin{equation*}
\mathcal{V}^{-}(X):=\left\{
V\in\mathbb{R}^{6}\,:\, V=\lim_{t\rightarrow 0-}(\Pi_{2}\circ T_{t})(Z)\hspace{2mm}\text{for}\hspace{2mm}Z=[X, V]
\right\}.
\end{equation*}
Similarly, the set of velocities $\mathcal{V}^{+}(X)\subseteq\mathbb{R}^{6}$ which are {\bf post-collisional} with respect to the collision configuration $X$ is
\begin{equation*}
\mathcal{V}^{+}(X):=\left\{
V\in\mathbb{R}^{6}\,:\, V=\lim_{t\rightarrow 0-}(L\circ \Pi_{2}\circ T_{-t})(Z)\hspace{2mm}\text{for}\hspace{2mm}Z=[X, V]
\right\}.
\end{equation*}
\end{defn}
We now specify the class of flows on $\mathscr{D}_{2}(\mathsf{P}_{\ast})$ with which we work in the rest of this article.
\begin{defn}[Regular Flow]\label{rbf}
We call a 1-parameter group $\{T_{t}\}_{t\in\mathbb{R}}$ of maps a {\bf regular flow} on $\mathscr{D}_{2}(\mathsf{P}_{\ast})$ if and only if the following properties hold true:
\begin{enumerate}[(R1)]
\item For each $Z_{0}\in\mathscr{D}_{2}(\mathsf{P}_{\ast})$, the set of collision times $\mathcal{T}(Z_{0})$ is either (i) the empty set $\varnothing$, (ii) the whole real line $\mathbb{R}$, or (iii) a locally-finite countable subset of $\mathbb{R}$; \vspace{1mm}
\item For each $Z_{0}\in\mathscr{D}_{2}(\mathsf{P}_{\ast})$, the map $t\mapsto (\Pi_{1}\circ T_{t})(Z_{0})$ lies in $C^{1}_{-}(\mathbb{R}, \mathbb{R}^{6})\cap C^{1}_{+}(\mathbb{R}, \mathbb{R}^{6})$. Moreover, the map $t\mapsto (\Pi_{2}\circ T_{t})(Z_{0})$ lies in $\mathrm{LSC}(\mathbb{R}, \mathbb{R}^{6})\cap C_{-}^{1}(\mathbb{R}, \mathbb{R}^{6})$. Finally, both the maps $t\mapsto (\Pi_{1}\circ T_{t})(Z_{0})$ and $t\mapsto (\Pi_{2}\circ T_{t})(Z_{0})$ are real analytic on the open set $\mathbb{R}\setminus\mathcal{T}(Z_{0})$; \vspace{1mm}
\item For almost every $X\in\partial\mathscr{P}_{2}(\mathsf{P}_{\ast})$, $\mathcal{V}^{-}(X)$ and $\mathcal{V}^{+}(X)$ are homeomorphic to the closed half-space $\{V\in\mathbb{R}^{6}\,:\,V_{6}\geq 0\}$; \vspace{1mm}
\item For almost every $X\in\partial\mathscr{P}_{2}(\mathsf{P}_{\ast})$, the map $s(\cdot; X):\mathcal{V}^{-}(X)\rightarrow\mathcal{V}^{+}(X)$ defined by
\begin{equation*}
s(\cdot; X):V\mapsto \lim_{t\rightarrow 0+}(\Pi_{2}\circ T_{t})(Z) \quad \text{for}\hspace{2mm}Z=[X, V]
\end{equation*}
is a Lebesgue measure-preserving $C^{1}$ diffeomorphism; \vspace{1mm}
\item For each $t\in\mathbb{R}$, if $T_{t}^{R}:\mathscr{D}_{2}(\mathsf{P}_{\ast})\rightarrow\mathscr{D}_{2}(\mathsf{P}_{\ast})$ is the operator defined by the relations
\begin{equation*}
\Pi_{1}\circ T_{t}^{R}:=\Pi_{1}\circ T_{-t} \quad \text{and}\quad \Pi_{2}\circ T_{t}^{R}:= L\circ\Pi_{2}\circ T_{-t},
\end{equation*} 
then $\{T_{t}^{R}\}_{t\in\mathbb{R}}=\{T_{t}\}_{t\in\mathbb{R}}$.
\end{enumerate}
\end{defn}
\begin{rem}
We draw the reader's attention to the fact that properties (R3) and (R4) are formulated for only {\em almost every} boundary point of $\mathscr{P}_{2}(\mathsf{P}_{\ast})$ as opposed to all points thereof. The reason for this choice is that there are phase spaces $\mathscr{D}_{2}(\mathsf{P}_{\ast})$ generated by `reasonable' strictly convex sets $\mathsf{P}_{\ast}$ and `natural' regular flows defined thereon for which there exist points $X\in\partial\mathscr{P}_{2}(\mathsf{P}_{\ast})$ such that $\mathcal{V}^{-}(X)$ is {\em not} a closed half-space in $\mathbb{R}^{6}$. Indeed, it has been shown in \textsc{Palffy-Muhoray, Virga, Wilkinson and Zheng} \cite{palffy2017paradox} and proved in \textsc{Wilkinson} \cite{mw111} that when $\mathsf{P}_{\ast}$ is taken to be an ellipse, that ``almost every'' cannot be replaced by ``every''.
\end{rem}
The class of regular flows is readily seen to be non-empty class (consider the case when $\mathsf{P}_{\ast}$ is a disk, for instance). Let us now comment briefly on the meaning of each of the above properties.
\subsection{A Brief Discussion of the Properties (R1)--(R6)}

\subsubsection{(R1) and (R2): Trajectory Regularity}
As one may readily check, this property is certainly consistent with the case of two congruent hard disks evolving in the whole space $\mathbb{R}^{2}$. In the case that the reference particle $\mathsf{P}_{\ast}\subset\mathbb{R}^{2}$ does not possess rotational symmetry, these regularity criteria simply reduce to the exclusion of (i) the phenomenon of accumulation of collision times of two such particles on any compact time interval, and (ii) rolling of one particle over another in free space.
\subsubsection{(R3): Scattering Map Regularity}
Whilst perhaps unimportant for the construction of weak solutions of \eqref{nemo}, this condition is required from the point of view of the theory of the Boltzmann equation. Indeed, the reader can verify that for any unit vector $n\in\mathbb{R}^{3}$, the $C^{1}$-diffeomorphism of $\mathbb{R}^{6}$ effected by the map
\begin{equation}\label{velmappy}
\left[
\begin{array}{c}
v \\
\ov{v}
\end{array}
\right]\mapsto \left[
\begin{array}{c}
v-n\otimes n(v-\ov{v}) \\
\ov{v}+n\otimes n(v-\ov{v})
\end{array}
\right]
\end{equation}
has unit Jacobian on $\mathbb{R}^{6}$, and that it is precisely this property of the scattering matrix that allows one to prove the well-known H-theorem for Boltzmann's kinetic equation: see \textsc{Cercignani, Illner and Pulvirenti} (\cite{cercignani2013mathematical}, chapter 3) for further details on this point.
\subsubsection{(R4): The Dynamics is Time-reversible}
It is a formal convention in kinetic theory that the family of solution operators associated to the ODEs of $N$-particle motion admit the property that they `cannot discern past from future'. Thus, property (R4) is simply a mathematical articulation of this convention. On the other hand, owing to the H-theorem, it is well known that this property does not hold for the formal flow of solution operators for the average dynamics as governed by the Boltzmann equation.

\subsection{\label{sec:level1}Weak Solutions and Physical Regular Flows}
We now define some physical functionals (derived from Euler's First and Second Laws of Motion, c.f. \textsc{Truesdell} \cite{truesdell}) of the dynamics generated by a regular flow $\{T_{t}\}_{t\in\mathbb{R}}$. We denote by $\mathrm{LM}:\mathbb{R}\times\mathscr{D}_{2}(\mathsf{P}_{\ast})\rightarrow\mathbb{R}^{2}$ the {\em linear momentum functional} given by
\begin{equation}\label{lmf}
\mathrm{LM}(t, Z_{0}):=m\left(
\begin{array}{c}
(\Pi_{2}\circ T_{t})(Z_{0})_{1}+(\Pi_{2}\circ T_{t})(Z_{0})_{3}\\
(\Pi_{2}\circ T_{t})(Z_{0})_{2}+(\Pi_{2}\circ T_{t})(Z_{0})_{4}
\end{array}
\right).
\end{equation}
We also define the {\em angular momentum functional} $\mathrm{AM}:\mathbb{R}^{2}\times \mathbb{R}\times\mathscr{D}_{2}(\mathsf{P}_{\ast})\rightarrow\mathbb{R}$ by
\begin{align}\label{amf}
\mathrm{AM}(a, t, Z_{0}):=
-m\left(
\begin{array}{c}
a_{1}-(\Pi_{1}\circ T_{t})(Z_{0})_{1} \\
a_{2}-(\Pi_{1}\circ T_{t})(Z_{0})_{2}
\end{array}
\right)^{\perp}\cdot \left(
\begin{array}{c}
(\Pi_{2}\circ T_{t})(Z_{0})_{1} \\
(\Pi_{2}\circ T_{t})(Z_{0})_{2}
\end{array}
\right)+J(\Pi_{2}\circ T_{t})(Z_{0})_{5}\notag \\ -m\left(
\begin{array}{c}
a_{1}-(\Pi_{1}\circ T_{t})(Z_{0})_{3} \\
a_{2}-(\Pi_{1}\circ T_{t})(Z_{0})_{4}
\end{array}
\right)^{\perp}\cdot \left(
\begin{array}{c}
(\Pi_{2}\circ T_{t})(Z_{0})_{3} \\
(\Pi_{2}\circ T_{t})(Z_{0})_{4}
\end{array}
\right)+J(\Pi_{2}\circ T_{t})(Z_{0})_{6}.
\end{align}
Finally, we write $\mathrm{KE}:\mathbb{R}\times\mathscr{D}_{2}(\mathsf{P}_{\ast})\rightarrow\mathbb{R}$ to denote the {\em kinetic energy functional} given by
\begin{equation}\label{kef}
\mathrm{KE}(t, Z_{0}):=|M(\Pi_{2}\circ T_{t})(Z_{0})|^{2},
\end{equation}
where $M\in\mathbb{R}^{6\times 6}$ is the mass-inertia matrix. We are finally in a position to define what we mean by a `physical' regular flow on $\mathscr{D}_{2}(\mathsf{P}_{\ast})$.
\begin{defn}[Physical Regular Flow]\label{physregflow}
A {\bf physical regular flow} $\{T_{t}\}_{t\in\mathbb{R}}$ on $\mathscr{D}_{2}(\mathsf{P}_{\ast})$ is a regular flow for which every choice of initial datum $Z_{0}\in\mathscr{D}_{2}(\mathsf{P}_{\ast})$, the trajectory $t\mapsto Z(t):=T_{t}Z_{0}$ is a global-in-time physical weak solution of \eqref{newty}. In particular, it respects the conservation of linear momentum
\begin{equation}\label{lm}
\mathrm{LM}(t, Z_{0})=\mathrm{LM}(0, Z_{0})\quad \text{for all}\hspace{2mm}t\in\mathbb{R},
\end{equation}
the conservation of angular momentum
\begin{equation}\label{am}
\mathrm{AM}(a, t, Z_{0})=\mathrm{AM}(a, 0, Z_{0})\quad \text{for all}\hspace{2mm}t\in\mathbb{R}
\end{equation}
and any $a\in\mathbb{R}^{3}$, and the conservation of kinetic energy
\begin{equation}\label{ke}
\mathrm{KE}(t, Z_{0})=\mathrm{KE}(0, Z_{0})\quad \text{for all}\hspace{2mm}t\in\mathbb{R}.
\end{equation}
\end{defn}
Owing to the general existence theory of \textsc{Ballard} \cite{ball}, it will prove useful to reduce our study of physical regular flows on $\mathscr{D}_{2}(\mathsf{P}_{\ast})$ to the family of {\em scattering maps} to which they give rise. We do this in the following section.
\section{\label{sec:level1}Physical Regular Flows and their Scattering Maps}\label{scat}
In this section, we derive the basic properties of scattering maps associated to any given physical regular flow. As we have already seen in property (R4) of regular flows above, scattering maps are an injective assignment of post-collisional velocities (both linear and angular) to given pre-collisional velocities. In order to study scattering maps, it will prove useful parameterise the set of collision configurations of two congruent particles $\partial\mathscr{P}_{2}(\mathsf{P}_{\ast})$ in a convenient manner. 
\subsection{An Atlas for the Set of Collision Configurations}\label{para}
In the case that $\mathsf{P}_{\ast}$ is a compact, strictly-convex set with real-analytic boundary, we claim that the boundary set $\partial\mathscr{P}_{2}(\mathsf{P}_{\ast})$ itself admits the structure of a real-analytic manifold with boundary\footnote{This is not the case for $N\geq 3$ hard particles, in which case $\partial\mathcal{P}_{3}(\mathsf{P}_{\ast})$ is seen readily only to be real analytic manifold with corners.}. To show this, we construct an atlas for $\partial\mathscr{P}_{2}(\mathsf{P}_{\ast})$ comprised of real-analytic charts. 

Let a boundary point $X=[x, \ov{x}, \vartheta, \ov{\vartheta}]\in\partial\mathscr{P}_{2}(\mathsf{P}_{\ast})$ be given. Such an $X$ on the boundary is of the shape
\begin{equation*}
X=\left(
\begin{array}{c}
x_{1} \\
x_{2} \\
x_{1}+d_{\beta(X)}\cos\psi(X) \\
x_{2}+d_{\beta(X)}\sin\psi(X) \\
\alpha(X) \\
\ov{\alpha}(X)
\end{array}
\right)
\end{equation*}
for some $x\in\mathbb{R}^{2}$ and $\beta(X):=(\alpha(X), \ov{\alpha}(X), \psi(X))\in\mathbb{T}^{3}$ where the number $d_{\beta}>0$ denotes the so-called {\em distance of closest approach} given by \begin{equation}\label{dca}
d_{\beta}:=\inf\left\{d>0\,:\,\mathrm{card}\,(R(\vartheta)\mathsf{P}_{\ast}\cap (R(\ov{\vartheta})\mathsf{P}_{\ast}+d e(\psi)))=0\right\},
\end{equation}
for $\beta=(\vartheta, \ov{\vartheta}, \psi)\in\mathbb{T}^{3}$. As the distance of closest approach is invariant with respect to global rotations of the particle system (i.e. $d_{\beta}=d_{\beta+\beta'}$ for all $\beta, \beta'\in\mathbb{T}^{3}$), one can check that
\begin{equation*}
d_{\beta}=D(\ov{\vartheta}-\vartheta, \psi-\vartheta)\quad \text{for all}\hspace{2mm}\beta=(\vartheta, \ov{\vartheta}, \psi)\in\mathbb{T}^{3},
\end{equation*}
where
\begin{equation*}
D(\theta, \psi):=\inf\left\{d>0\,:\,\mathrm{card}\,(\mathsf{P}_{\ast}\cap (R(\ov{\theta})\mathsf{P}_{\ast}+d e(\psi)))=0\right\}.
\end{equation*}
If $N_{\beta(X)}\subset\mathbb{T}^{3}$ denotes some open neighbourhood of $\beta(X)$, we define the associated local chart map $\phi_{X}: \mathbb{R}^{2}\times N_{\beta(X)}\rightarrow\partial\mathscr{P}_{2}(\mathsf{P}_{\ast})$ by
\begin{equation*}
\phi_{X}(\xi):=\left(
\begin{array}{c}
\xi_{1}\\
\xi_{2} \\
d_{\beta}\cos\psi \\
d_{\beta}\sin\psi \\
\alpha \\
\ov{\alpha}
\end{array}
\right) \quad \text{for}\hspace{2mm}(\xi_{1}, \xi_{2}, \alpha, \ov{\alpha}, \psi)\in\mathbb{R}^{2}\times N_{\beta(X)}.
\end{equation*}
It can be shown (see \textsc{Palffy-Muhoray and Zheng} \cite{zheng2007distance} in the case of ellipses, for instance) that $\phi_{X}$ is a real analytic map on $\mathbb{R}^{2}\times N_{\beta(X)}$. As such, the family $\Phi:=\{(\phi_{X}, N_{\beta(X)})\}_{X\in \partial\mathscr{P}_{2}(\mathsf{P}_{\ast})}$ admits the property of being an atlas for $\partial\mathscr{P}_{2}(\mathsf{P}_{\ast})$ comprised of real-analytic maps. 

As it will be of use in the sequel, we calculate the outward normal map to $\partial\mathscr{P}_{2}(\mathsf{P}_{\ast})$ explicitly in terms of the distance of closest approach map $\beta\mapsto d_{\beta}$ on $\mathbb{T}^{3}$. In loose terms, the outward normal for us will be the unit vector that `points into' the region $\mathbb{R}^{6}\setminus\mathscr{P}_{2}(\mathsf{P}_{\ast})$ at points on the boundary $\partial\mathscr{P}_{2}(\mathsf{P}_{\ast})$. We have the following lemma.
\begin{lem}
For $\beta=(\vartheta, \ov{\vartheta}, \psi)\in\mathbb{T}^{3}$, let $X_{\beta}\in\partial\mathscr{P}_{2}(\mathsf{P}_{\ast})$ denote the point
\begin{equation*}
X_{\beta}:=\left(
\begin{array}{c}
0 \\
0 \\
d_{\beta}\cos\psi \\
d_{\beta}\sin\psi \\
\vartheta \\
\ov{\vartheta}
\end{array}
\right).
\end{equation*}
The outward unit normal $\widehat{\gamma}_{\beta}\in\mathbb{S}^{5}$ to $\partial\mathscr{P}_{2}(\mathsf{P}_{\ast})$ at $X_{\beta}$ is given by
\begin{equation}\label{bigolevector}
\widehat{\gamma}_{\beta}:=\frac{1}{\Lambda_{\beta}}\left[
\begin{array}{c}
\displaystyle \frac{-e(\psi)+\displaystyle\frac{1}{d_{\beta}}\displaystyle\frac{\partial D}{\partial \theta}(\ov{\vartheta}-\vartheta, \psi-\vartheta)e(\psi)^{\perp}}{1+\displaystyle\frac{1}{d_{\beta}^{2}}\left(\displaystyle \frac{\partial D}{\partial \theta}(\ov{\vartheta}-\vartheta, \psi-\vartheta)\right)^{2}} \vspace{2mm}\\
\displaystyle \frac{e(\psi)-\displaystyle\frac{1}{d_{\beta}}\displaystyle\frac{\partial D}{\partial \theta}(\ov{\vartheta}-\vartheta, \psi-\vartheta)e(\psi)^{\perp}}{1+\displaystyle\frac{1}{d_{\beta}^{2}}\left(\displaystyle \frac{\partial D}{\partial \theta}(\ov{\vartheta}-\vartheta, \psi-\vartheta)\right)^{2}} \vspace{2mm} \\
\displaystyle -\frac{\displaystyle \frac{\partial D}{\partial \theta}(\ov{\vartheta}-\vartheta, \psi-\vartheta)+\frac{\partial D}{\partial \psi}(\ov{\vartheta}-\vartheta, \psi-\vartheta)}{1+\displaystyle\frac{1}{d_{\beta}^{2}}\left(\displaystyle \frac{\partial D}{\partial \theta}(\ov{\vartheta}-\vartheta, \psi-\vartheta)\right)^{2}} \vspace{2mm} \\
\displaystyle \frac{\displaystyle \frac{\partial D}{\partial \theta}(\ov{\vartheta}-\vartheta, \psi-\vartheta)}{1+\displaystyle\frac{1}{d_{\beta}^{2}}\left(\displaystyle \frac{\partial D}{\partial \theta}(\ov{\vartheta}-\vartheta, \psi-\vartheta)\right)^{2}}
\end{array}
\right]
\end{equation}
where $\Lambda_{\beta}>0$ is the normalisation factor which renders $\|\widehat{\gamma}_{\beta}\|=1$. 
\end{lem}
\begin{proof}
This follows from a calculation, using the fact that $X_{\beta}$ lies on the 0-level set of the function $F:\mathbb{R}^{6}\rightarrow\mathbb{R}$ given by
\begin{equation*}
F(X):=|x-\ov{x}|-d_{\beta(X)} \quad \text{for}\hspace{2mm}X\in\partial\mathscr{P}_{2}(\mathsf{P}_{\ast}),
\end{equation*}
and that $D$ inherits the regularity of the boundary of $\mathsf{P}_{\ast}$.
\end{proof}
It is at this point that one can see the pertinence of the half-spaces involved in defining phase space $\mathscr{D}_{2}(\mathsf{P}_{\ast})$ as a fibre bundle. It is fair to claim that expression \eqref{bigolevector} is somewhat unwieldy. Thankfully, we can characterise the normal $\widehat{\gamma}_{\beta}\in\mathbb{S}^{5}$ to the codimension 1 set $\partial\mathscr{P}_{2}(\mathsf{P}_{\ast})\subset\mathbb{R}^{6}$ in terms of more familiar `collision data' at the level of $\mathbb{R}^{2}$. The collision data of which we speak are the {\em collision vector} $p_{\beta}\in\mathbb{R}^{2}$, defined as the unique element of the singleton set
\begin{equation*}
R(\vartheta)\mathsf{P}_{\ast}\cap(R(\ov{\vartheta})\mathsf{P}_{\ast}+d_{\beta}e(\psi)),
\end{equation*}
together with the {\em conjugate collision vector} $q_{\beta}:=p_{\beta}-d_{\beta}e(\psi)$ and the outward {\em contact normal} $n_{\beta}\in\mathbb{S}^{1}$ to the closed curve $R(\vartheta)\partial\mathsf{P}_{\ast}$ at the point $p_{\beta}\in R(\vartheta)\partial\mathsf{P}_{\ast}$. Indeed, it will be shown that
\begin{equation*}
p_{\beta}^{\perp}\cdot n_{\beta}=\displaystyle \frac{\displaystyle \frac{\partial D}{\partial \theta}(\ov{\vartheta}-\vartheta, \psi-\vartheta)+\frac{\partial D}{\partial \psi}(\ov{\vartheta}-\vartheta, \psi-\vartheta)}{1+\displaystyle\frac{1}{d_{\beta}^{2}}\left(\displaystyle \frac{\partial D}{\partial \theta}(\ov{\vartheta}-\vartheta, \psi-\vartheta)\right)^{2}} 
\end{equation*}
and
\begin{equation*}
n_{\beta}=\displaystyle \frac{e(\psi)-\displaystyle\frac{1}{d_{\beta}}\displaystyle\frac{\partial D}{\partial \theta}(\ov{\vartheta}-\vartheta, \psi-\vartheta)e(\psi)^{\perp}}{1+\displaystyle\frac{1}{d_{\beta}^{2}}\left(\displaystyle \frac{\partial D}{\partial \theta}(\ov{\vartheta}-\vartheta, \psi-\vartheta)\right)^{2}},
\end{equation*}
for $\beta\in\mathbb{T}^{3}$. We believe that this is the first time these identities have been established in works on the geometry of smooth convex bodies. In any case, as a result of the above identities, $p_{\beta}$, $q_{\beta}$ and $n_{\beta}$ constitute the essential spatial data one employs to construct boundary conditions for the IBVP associated to \eqref{nemo} in the sequel.
\subsection{Pre- and Post-collisional Velocities}
As has been discussed in \textsc{Wilkinson} \cite{mw111}, the structure of the set of pre- and post-collisional velocities depends sensitively on geometric properties of $\mathsf{P}_{\ast}$. Aside from the case when $\mathsf{P}_{\ast}$ is a disk, it is in general a very difficult task to characterise explicitly, for all $X\in\partial\mathscr{P}_{2}(\mathsf{P}_{\ast})$, the sets $\mathcal{V}^{+}(X)$ and $\mathcal{V}^{-}(X)$ for physical regular flows. For the purposes of demonstrating our main non-uniqueness result, however, it is not necessary to obtain a characterisation.

It is important to emphasise the fact that, due to the method of proof we employ, we work only with those subsets of $\mathbb{R}^{2}$ which are realised as the level set of real-analytic functions on $\mathbb{R}^{2}$. The following proof has been taken, in essence, from \textsc{Pallfy-Muhoray, Virga, Wilkinson and Zheng} \cite{palffy2017paradox}.
\begin{prop}\label{taylorstuff}
Let $\{T_{t}\}_{t\in\mathbb{R}}$ be a physical regular flow on $\mathscr{P}_{2}(\mathsf{P}_{\ast})$, where $\mathsf{P}_{\ast}\subset\mathbb{R}^{2}$ is a compact, strictly convex set with the following properties: 
\begin{enumerate}[{\rm (B1)}]

\item $\partial\mathsf{P}_{\ast}$ admits the structure of a 1-dimensional real-analytic manifold; and 
\item there exists a real-analytic map $b_{\ast}:\mathbb{R}^{2}\rightarrow\mathbb{R}$ such that
\begin{itemize} 
\item $\partial\mathsf{P}_{\ast}=\{y\in\mathbb{R}^{2}\,:\,b_{\ast}(y)=0\}$;
\item $\mathrm{int}\,\mathsf{P}_{\ast}=\{y\in\mathbb{R}^{2}\,:\,b_{\ast}(y)<0\}$;
\item $\mathbb{R}^{2}\setminus\mathsf{P}_{\ast}=\{y\in\mathbb{R}^{2}\,:\,b_{\ast}(y)>0\}$.
\end{itemize}
\end{enumerate}
Then, for any $X\in\partial\mathscr{P}_{2}(\mathsf{P}_{\ast})$ one has that
\begin{equation*}
\mathcal{V}^{-}(X)\supseteq\left\{V\in\mathbb{R}^{6}\,:\,V\cdot M\widehat{\nu}_{\beta}< 0\right\} 
\end{equation*}
and\begin{equation}\label{postee}
\mathcal{V}^{+}(X)\supseteq\left\{V\in\mathbb{R}^{6}\,:\,V\cdot M\widehat{\nu}_{\beta}> 0\right\},
\end{equation}
where $\widehat{\nu}_{\beta}\in\mathbb{R}^{6}$ is the unit vector
\begin{equation}\label{eetah}
\widehat{\nu}_{\beta}:=\frac{1}{\sqrt{\frac{2}{m}+\frac{1}{J}|p_{\beta}^{\perp}\cdot n_{\beta}|^{2}+\frac{1}{J}|q_{\beta}^{\perp}\cdot n_{\beta}|^{2}}}M^{-1}\left[
\begin{array}{c}
-n_{\beta} \\ 
n_{\beta} \\ 
-p_{\beta}^{\perp}\cdot n_{\beta} \\ 
q_{\beta}^{\perp}\cdot n_{\beta}
\end{array}
\right].
\end{equation}
\end{prop}
\begin{proof}
Let $Z_{0}\in\partial\mathscr{P}_{2}(\mathsf{P}_{\ast})$ be given. We begin by defining the map $b:\mathbb{R}^{2}\times\mathbb{R}\rightarrow\mathbb{R}$ by
\begin{equation*}
b(y, t):=b_{\ast}(R(\vartheta(t))^{T}(y-x(t))) 
\end{equation*}
for $y\in\mathbb{R}^{2}$ and $t\in\mathbb{R}$, where $\vartheta(t)$ and $x(t)$ are determined by the map $t\mapsto T_{t}Z_{0}$. By definition of $b_{\ast}$, one has that
\begin{align}
b(y, t)\left\{
\begin{array}{ll}
<0 & \quad \text{if}\hspace{2mm} y\in\mathrm{int}\,\mathsf{P}(t), \vspace{2mm}\\
=0 & \quad \text{if} \hspace{2mm}y\in\partial\mathsf{P}(t), \vspace{2mm} \\
>0 & \quad \text{if} \hspace{2mm}y\notin\mathsf{P}(t).
\end{array}
\right.
\end{align}
As such, one can consider $b$ as a marker for the dynamics. Indeed, if $\{T_{t}\}_{t\in\mathbb{R}}$ is a regular flow, then $b(\xi(t), t)\geq 0$ for the trajectory $t\mapsto \xi(t)$ of any material point on $\ov{\mathsf{P}}(t)$. Suppose, without loss of generality, that $Z_{0}\in\partial\mathscr{P}_{2}(\mathsf{P}_{\ast})$ has the property that $\vartheta_{0}=0$. The trajectory $t\mapsto T_{t}Z_{0}$ is real analytic in a left neighbourhood $I^{-}(Z_{0}):=(-\delta, 0)$ for some $\delta=\delta(Z_{0})>0$. We study the motion of the points $P=P(t)$ and $Q=Q(t)$ on the sets $\mathsf{P}(t)$ and $\ov{\mathsf{P}}(t)$, respectively, whose associated displacement vectors are defined by
\begin{equation}\label{kyoo}
x_{Q}(t):=\ov{x}(t)+R(\ov{\vartheta}(t))R(\ov{\vartheta}_{0})^{T}q_{\beta} \quad \text{and}\quad x_{P}(t):=x(t)+R(\vartheta(t))p_{\beta}.
\end{equation}
One has by Taylor's theorem that
\begin{equation*}
b_{\ast}(y)=\nabla b_{\ast}(p_{\beta})\cdot (y-p_{\beta})+\frac{1}{2}(y-p_{\beta})\cdot D^{2}b_{\ast}(p_{\beta})(y-p_{\beta})+ \mathcal{O}(|y-p_{\beta}|^{3})
\end{equation*}
as $y\rightarrow p_{\beta}$. In turn, one can show that $b(x_{Q}(t), t)$ admits the expansion
\begin{align}
b(x_{Q}(t), t)= t\left(n_{\beta}\cdot (\dot{x}_{Q}(0)-\dot{x}_{P}(0))\right) \vspace{2mm}\notag \\
+\frac{t^{2}}{2}\bigg(2\omega_{0}\cdot n_{\beta}^{\perp}\cdot(\dot{x}_{Q}(0)-\dot{x}_{P}(0))+n_{\beta}\cdot (\ddot{x}_{Q}(0)-\ddot{x}_{P}(0)) \vspace{2mm} \notag \\
+(\dot{x}_{Q}(0)-\dot{x}_{P}(0))\cdot D^{2}b_{\ast}(p_{\beta})(\dot{x}_{Q}(0)-\dot{x}_{P}(0))\bigg)+\mathcal{O}(t^{3})\quad \text{as}\hspace{2mm}t\rightarrow 0-,
\end{align}
where
\begin{equation*}
\dot{x}_{Q}(0)-\dot{x}_{P}(0)=\ov{v}_{0}+\ov{\omega}_{0}q_{\beta}^{\perp}-v_{0}-\omega_{0}p_{\beta}^{\perp}
\end{equation*}
and
\begin{equation*}
\ddot{x}_{Q}(0)-\ddot{x}_{P}(0)=-\ov{\omega}_{0}^{2}q_{\beta}+\omega_{0}^{2}p_{\beta}.
\end{equation*}
Manifestly, on a set of velocities of full measure, it is the sign of the leading coefficient $n_{\beta}\cdot (\dot{x}_{Q}(0)-\dot{x}_{P}(0))$ that determines the behaviour of the point trajectory $t\mapsto x_{Q}(t)$ for $t$ in a left neighbourhood of 0. As the coefficient of $t$ in the expansion of $b(x_{Q}(t), t)$ determines its sign in a left neighbourhood of $t=0$, it follows that if $n_{\beta}\cdot(v_{Q}(0)-v_{P}(0))<0$ then $b(x_{Q}(t), t)\geq 0$ for $t\in (-\delta', 0]$. Thus, we infer that
\begin{equation*}
\{V\in\mathbb{R}^{6}\,:\, M\widehat{\nu}_{\beta}\cdot V<0\}\subseteq\mathcal{V}^{-}(X),
\end{equation*}
where $\widehat{\nu}_{\beta}$ is the unit vector defined in \eqref{eetah} above. The case for post-collisional velocities is handled in the same way, yielding \eqref{postee}.
\end{proof}
We now prove the above-claimed geometric identities for the collision data in terms of the distance of closest approach map.
\begin{prop}
Suppose $\mathsf{P}_{\ast}\subset\mathbb{R}^{2}$ satisfies the same hypotheses as in Proposition \ref{taylorstuff} above. For any $\beta\in\mathbb{T}^{3}$, one has that
\begin{equation}\label{pee}
p_{\beta}^{\perp}\cdot n_{\beta}=\displaystyle \frac{\displaystyle \frac{\partial D}{\partial \theta}(\ov{\vartheta}-\vartheta, \psi-\vartheta)+\frac{\partial D}{\partial \psi}(\ov{\vartheta}-\vartheta, \psi-\vartheta)}{1+\displaystyle\frac{1}{d_{\beta}^{2}}\left(\displaystyle \frac{\partial D}{\partial \theta}(\ov{\vartheta}-\vartheta, \psi-\vartheta)\right)^{2}} 
\end{equation}
and
\begin{equation}\label{enn}
n_{\beta}=\displaystyle \frac{e(\psi)-\displaystyle\frac{1}{d_{\beta}}\displaystyle\frac{\partial D}{\partial \theta}(\ov{\vartheta}-\vartheta, \psi-\vartheta)e(\psi)^{\perp}}{1+\displaystyle\frac{1}{d_{\beta}^{2}}\left(\displaystyle \frac{\partial D}{\partial \theta}(\ov{\vartheta}-\vartheta, \psi-\vartheta)\right)^{2}},
\end{equation}
\end{prop}
\begin{proof}
To establish the geometric identities \eqref{pee} and \eqref{enn}, we employ an argument which uses properties of dynamics. Specifically, we employ the existence theory of \textsc{Ballard}. By  \cite{ball}, there exists a kinetic energy-conserving regular flow $\{E_{t}^{}\}_{t\in\mathbb{R}}$ on $\mathscr{P}_{2}(\mathsf{P}_{\ast})$. This regular flow admits the property that if $X_{0}\in\partial\mathscr{P}_{2}(\mathsf{P}_{\ast})$ and $V_{0}\in\Sigma_{\beta}^{-}$ (where $\beta\in\mathbb{T}^{3}$ is determined by $X_{0}$), then the map $t\mapsto X(t):=(\Pi_{1}\circ E_{t})(Z_{0})\in C^{1}_{-}(\mathbb{R}, \mathbb{R}^{6})\cap C^{1}_{+}(\mathbb{R}, \mathbb{R}^{6})$ has the property that
\begin{equation*}
F(X(t))\geq 0 \quad \text{in a left neighbourhood of 0},
\end{equation*}
and $F(X(0))=0$. Furthermore, one can check that
\begin{equation*}
b(x_{Q}(t), t)\geq 0 \quad \Longleftrightarrow\quad F(X(t))\leq 0,
\end{equation*}
where $x_{Q}:\mathbb{R}\rightarrow\mathbb{R}^{2}$ is given in \eqref{kyoo} above. By Proposition \ref{taylorstuff}, for any $W_{0}\in\{V\in\mathbb{R}^{6}\,:\,M\widehat{\nu}_{\beta}\cdot V<0\}$, the associated $t\mapsto X_{1}(t):=(\Pi\circ E_{t})(\zeta_{0})$ satisfies $b(x_{Q}(t), t)\geq 0$ for $t$ in a sufficiently-small left neighbourhood of 0, whence
\begin{equation*}
\frac{d}{dt_{-}}F(X_{1}(t))|_{t=0}\leq 0.
\end{equation*}
As such, the given $W_{0}$ satisfies the inequality $V_{0}\cdot\widehat{\gamma}_{\beta}\leq 0$. This is only possible if the open half spaces in $\mathbb{R}^{6}$ coincide with one another. This holds if and only if $M\widehat{\nu}_{\beta}=\widehat{\gamma}_{\beta}$ in $\mathbb{S}^{5}$ for all $\beta\in\mathbb{T}^{2}$, and so we are done.
\end{proof}
We shall use the result of this Proposition when defining the so-called {\em Monge-Amp\`{e}re scattering problem}, a constrained PDE problem for the systematic derivation of {\em scattering maps}, which we define formally in the following section.
\subsection{From Regular Flows to Scattering Maps}
We now define the central object of interest in this work. Thereafter, we shall focus much of our attention deriving an appropriate set of PDE which governs it.
\begin{defn}\label{scatteringdefn}
Suppose $\{T_{t}\}_{t\in\mathbb{R}}$ is a regular flow on $\mathscr{D}_{2}(\mathsf{P}_{\ast})$. For a given $X\in\partial\mathscr{P}_{2}(\mathsf{P}_{\ast})$, its associated {\bf scattering map} $\sigma(\cdot; X):\mathcal{V}^{-}(X)\rightarrow\mathcal{V}^{+}(X)$ is defined pointwise by
\begin{equation}\label{scatmapdef}
\sigma(V; X)=\lim_{t\rightarrow 0+}(\Pi_{2}\circ T_{t})(Z) \quad \text{for}\hspace{2mm}V\in\mathcal{V}^{-}(X),
\end{equation}
where $Z:=[X, V]$.
\end{defn}
Note that as $\{T_{t}\}_{t\in\mathbb{R}}$ has the property (R2), it follows that $\sigma(\cdot; X):\mathcal{V}^{-}(X)\rightarrow\mathcal{V}^{+}(X)$ is a bijection. 
Due to the fact that weak solutions of Newton's equations of motion \eqref{newty} have the property that they are both translation invariant and rotation covariant, we have the following structural formula for scattering maps associated with regular flows:
\begin{lem}\label{symscat}
Suppose the regular flow $\{T_{t}\}_{t\in\mathbb{R}}$ has the property that $t\mapsto T_{t}Z_{0}$ is a global-in-time weak solution of \eqref{newty} for every $Z_{0}\in\mathscr{D}_{2}(\mathsf{P}_{\ast})$. For any $y\in\mathbb{R}^{2}$, let $\mathsf{T}_{y}:\mathscr{P}_{2}(\mathsf{P}_{\ast})\rightarrow\mathscr{P}_{2}(\mathsf{P}_{\ast})$ denote the translation operator 
\begin{equation*}
\mathsf{T}_{y}X:=\left[
\begin{array}{c}
x+y \\
\ov{x}+y \\
\vartheta \\
\ov{\vartheta}
\end{array}
\right]\quad \text{for}\hspace{2mm}X\in\mathscr{P}_{2}(\mathsf{P}_{\ast}),
\end{equation*}
and for any $\alpha\in\mathbb{S}^{1}$, let $\mathsf{R}_{\alpha}:\mathscr{P}_{2}(\mathsf{P}_{\ast})\rightarrow\mathscr{P}_{2}(\mathsf{P}_{\ast})$ denote the `rotation' operator
\begin{equation*}
\mathsf{R}_{\alpha}Z:=\left[
\begin{array}{c}
R(\alpha)x \\
R(\alpha)\ov{x} \\
R(\alpha)v \\
R(\alpha)\ov{v} \\
\vartheta+\alpha \\
\ov{\vartheta}+\alpha \\
\omega \\
\ov{\omega}
\end{array}
\right] \quad \text{for}\hspace{2mm}Z=[X, V].
\end{equation*}
The family of scattering maps $\{\sigma(\cdot; X)\}_{X\in\partial\mathscr{P}_{2}(\mathsf{P}_{\ast})}$ admits the following symmetries:
\begin{equation*}
\sigma(V; \mathsf{T}_{y}X)=\sigma(V; X) \quad \text{for all}\hspace{2mm}y\in\mathbb{R}^{2} 
\end{equation*}
and
\begin{equation*}
R(\alpha)^{T}\sigma((\Pi_{2}\circ\mathsf{R}_{\alpha})(Z); (\Pi_{1}\circ \mathsf{R}_{\alpha})(Z))=\sigma(V; X)\quad \text{for all}\hspace{2mm}\alpha\in\mathbb{S}^{1}.
\end{equation*}
\end{lem}
\begin{proof}
This is a simple exercise which is left to the reader.
\end{proof}
It follows quickly from the above observation that the scattering family $\{\sigma(\cdot; X)\}_{X\in\partial\mathscr{P}_{2}(\mathsf{P}_{\ast})}$ is generated by only 2 spatial parameters, namely $(\theta, \psi)\in\mathbb{T}^{2}$. To see this, we require the following definition.
\begin{defn}[Reference Scattering Family]
Let $\{T_{t}\}_{t\in\mathbb{R}}$ be a regular flow. Suppose $X\in\mathscr{P}_{2}(\mathsf{P}_{\ast})$ corresponds to the reference collision configuration given by
\begin{equation*}
X=\left(
\begin{array}{c}
0 \\
0 \\
d_{\theta}(\psi)\cos\psi \\
d_{\theta}(\psi)\sin\psi 
\end{array}
\right)
\end{equation*}
for some $(\theta, \psi)\in\mathbb{T}^{2}$. If $\mathcal{V}_{(\theta, \psi)}^{-}:=\mathcal{V}^{-}(X)$ and $\mathcal{V}_{(\theta, \psi)}^{+}:=\mathcal{V}^{+}(X)$, we define the {\bf reference scattering map} $\sigma_{(\theta, \psi)}:\mathcal{V}_{(\theta, \psi)}^{-}\rightarrow\mathcal{V}_{(\theta, \psi)}^{+}$ corresponding to $(\theta, \psi)$ by
\begin{equation*}
\sigma_{(\theta, \psi)}(V):=\sigma(V; X) \quad
\end{equation*}
for each $V\in\mathcal{V}_{(\theta, \psi)}^{-}$. The collection $\{\sigma_{(\theta, \psi)}\}_{(\theta, \psi)\in\mathbb{T}^{2}}$ is said to the {\em reference scattering family} associated to the regular flow $\{T_{t}\}_{t\in\mathbb{R}}$. 
\end{defn}
It is more convenient to work with the following in the sequel.
\begin{defn}[Canonical Scattering Family]\label{canon}
Let $\{T_{t}\}_{t\in\mathbb{R}}$ be a regular flow. Suppose $X\in\mathscr{P}_{2}(\mathsf{P}_{\ast})$ corresponds to the collision configuration given by
\begin{equation*}
X=\left(
\begin{array}{c}
0 \\
0 \\
d_{\beta}\cos\psi \\
d_{\beta}\sin\psi 
\end{array}
\right)
\end{equation*}
for some $\beta=(\vartheta, \ov{\vartheta}, \psi)\in\mathbb{T}^{3}$. If $\mathcal{V}_{\beta}^{-}:=\mathcal{V}^{-}(X)$ and $\mathcal{V}_{\beta}^{+}:=\mathcal{V}^{+}(X)$, we define the {\bf canonical scattering map} $\sigma_{\beta}:\mathcal{V}_{\beta}^{-}\rightarrow\mathcal{V}_{\beta}^{+}$ corresponding to $\beta$ by
\begin{equation*}
\sigma_{\beta}(V):=R(\vartheta)\sigma_{(\ov{\vartheta}-\vartheta, \psi)}(R(\vartheta)^{T}V)
\end{equation*}
for each $V\in\mathcal{V}_{\beta}^{-}$. The collection $\{\sigma_{\beta}\}_{\beta\in\mathbb{T}^{3}}$ is said to the {\em canonical scattering family} associated to the regular flow $\{T_{t}\}_{t\in\mathbb{R}}$. 
\end{defn}
With these definitions in place, we state the following corollary of lemma \ref{symscat}.
\begin{cor}
Suppose $\{T_{t}\}_{t\in\mathbb{R}}$ is a regular flow, and $\{\sigma(\cdot; X)\}_{X\in\partial\mathscr{P}_{2}(\mathsf{P}_{\ast})}$ its associated scattering family. For any $X\in\partial\mathscr{P}_{2}(\mathsf{P}_{\ast})$, it follows that
\begin{equation*}
\sigma(\cdot; X)=\sigma_{\beta} \quad \text{as maps on}\hspace{2mm}\mathcal{V}_{\beta}^{-},
\end{equation*}
where $\beta=\beta(X)$ is given by $(\vartheta, \ov{\vartheta}, \psi(X))$, where $\psi(X):=\mathrm{Arctan}[(\ov{x}_{2}-x_{2})/(\ov{x}_{1}-x_{1})]$.
\end{cor}
Together with the general existence theory of \textsc{Ballard} \cite{ball} for the dynamics of real-analytic hard particles, these observations allow us to reduce the study of regular flows to a study of their associated scattering maps.
\subsection{Some Properties of Canonical Scattering Families for Physical Regular Flows}
In what follows, we work in the smaller class of physical regular flows as defined in \ref{physregflow} above. We now gather some elementary properties of the canonical scattering family associated with a given physical regular flow.
\begin{rem}
In what follows, in line with the kinetic theory convention, post-collisional velocities and angular speeds in $\mathcal{V}^{+}_{\beta}$ shall be adorned with a prime $'$, with their pre-collisional counterparts in $\mathcal{V}^{-}_{\beta}$ remaining unprimed. In particular, $V_{\beta}'\in\mathcal{V}^{+}_{\beta}$ is written as $[v_{\beta}'. \ov{v}_{\beta}', \omega_{\beta}', \ov{\omega}_{\beta}']$.
\end{rem}
We appeal to the formalism of \textsc{Truesdell} \cite{truesdell}. Suppose we denote by $U:\mathbb{R}^{2}\times\mathbb{R}\times\mathscr{D}_{2}(\mathsf{P}_{\ast})\rightarrow\mathbb{R}^{2}$ the map that assigns a linear velocity to each material point in a body in space. Let $\{T_{t}\}_{t\in\mathbb{R}}$ be a physical regular flow on $\mathscr{D}_{2}(\mathsf{P}_{\ast})$, and let $Z_{0}\in\mathscr{D}_{2}(\mathsf{P}_{\ast})$ be given. At any collision time $\tau\in\mathcal{T}(Z_{0})$, by Euler's First Law we enforce the conservation of linear momentum
\begin{align}
\lim_{t\rightarrow \tau-}\int_{\mathsf{P}(t)}U(y, t; Z_{0})\,dy+\lim_{t\rightarrow\tau-}\int_{\ov{\mathsf{P}}(t)}U(y, t; Z_{0})\,dy \vspace{2mm}\notag \\
 =\lim_{t\rightarrow \tau+}\int_{\mathsf{P}(t)}U(y, t; Z_{0})\,dy+\lim_{t\rightarrow\tau+}\int_{\ov{\mathsf{P}}(t)}U(y, t; Z_{0})\,dy, \tag{COLM}
\end{align}
which reduces to
\begin{equation}\label{colm}
mv_{\beta}'+m\overline{v}_{\beta}'=mv+m\overline{v}.
\end{equation}
Let a `point of measurement' $a\in\mathbb{R}^{2}$ in the rigid set domain be given. Euler's Second Law enforces the conservation of angular momentum with respect to the point $a$ is written as
\begin{align}
\lim_{t\rightarrow \tau-}\int_{\mathsf{P}(t)}(y-a)^{\perp}\cdot U(y, t; Z_{0})\,dy+\lim_{t\rightarrow \tau-}\int_{\ov{\mathsf{P}}(t)}(y-a)^{\perp}\cdot U(y, t; Z_{0})\,dy \notag \vspace{2mm} \\
 =\lim_{t\rightarrow \tau+}\int_{\mathsf{P}(t)}(y-a)^{\perp}\cdot U(y, t; Z_{0})\,dy+\lim_{t\rightarrow \tau+}\int_{\ov{\mathsf{P}}(t)}(y-a)^{\perp}\cdot U(y, t; Z_{0})\,dy, \tag{COAM}
\end{align}
which by using \eqref{colm} reduces to
\begin{align}\label{coam}
J\omega_{\beta}'+md^{\overline{\vartheta}}_{\vartheta}(\psi)e(\psi))^{\perp}\cdot\overline{v}_{\beta}'+J\overline{\omega}_{\beta}' = J\omega+md^{\overline{\vartheta}}_{\vartheta}(\psi)e(\psi))^{\perp}\cdot\overline{v}+J\overline{\omega}.
\end{align}
Finally, Euler's Laws in the absence of external forces yields conservation of kinetic energy,
\begin{align}
\lim_{t\rightarrow \tau-}\frac{1}{2}\int_{\mathsf{P}(t)}|U(y, t; Z_{0})|^{2}\,dy+\lim_{t\rightarrow\tau-}\frac{1}{2}\int_{\ov{\mathsf{P}}(t)}|U(y, t; Z_{0})|^{2}\,dy \vspace{2mm}\\
 \lim_{t\rightarrow \tau+}\frac{1}{2}\int_{\mathsf{P}(t)}|U(y, t; Z_{0})|^{2}\,dy+\lim_{t\rightarrow\tau+}\frac{1}{2}\int_{\ov{\mathsf{P}}(t)}|U(y, t; Z_{0})|^{2}\,dy, \tag{COKE}
\end{align}
which reduces to
\begin{equation}\label{coke}
m|v_{\beta}'|^{2}+J(\omega_{\beta}')^{2}+m|\overline{v}_{\beta}'|^{2}+J(\overline{\omega}_{\beta}')^{2}
= m|v|^{2}+J\omega^{2}+m|\overline{v}|^{2}+J\overline{\omega}^{2}.
\end{equation}
It is now convenient to rewrite the conservation laws \eqref{colm}, \eqref{coam} and \eqref{coke} in the terms of scattering map notation for the physical regular flow $\{T_{t}\}_{t\in\mathbb{R}}$. Indeed, if $\sigma_{\beta}$ is to conserve total kinetic energy, then we have
\begin{equation}\label{scatcoke}
|M\sigma_{\beta}(V)|^{2}=|MV|^{2} \quad \text{for all}\hspace{2mm} V\in\mathbb{R}^{6},
\end{equation}
where $M\in\mathbb{R}^{6\times 6}$ is the mass-inertia matrix \eqref{mi} above. The conservation law \eqref{colm} can be recast in the form
\begin{equation}\label{scatcolm}
\widehat{E}_{1}\cdot \sigma_{\beta}(V)=\widehat{E}_{1}\cdot V \quad \text{and}\quad\widehat{E}_{2}\cdot \sigma_{\beta}(V)=\widehat{E}_{2}\cdot V\quad \text{for all}\hspace{2mm} V\in\mathbb{R}^{6},
\end{equation}
where $\widehat{E}_{1}, \widehat{E}_{2}\in\mathbb{R}^{6}$ are the unit vectors
\begin{equation}\label{eeee}
\widehat{E}_{1}:=\frac{1}{\sqrt{2}}\left[
\begin{array}{c}
1 \\ 0 \\ 1 \\ 0 \\ 0 \\ 0
\end{array}
\right]\quad \text{and}\quad\widehat{E}_{2}:=\frac{1}{\sqrt{2}}\left[
\begin{array}{c}
0 \\ 1 \\ 0 \\ 1 \\ 0 \\ 0
\end{array}
\right]
\end{equation}
Moreover, \eqref{coam} becomes
\begin{equation}\label{scatcoam}
\widehat{\Gamma}_{\beta}\cdot \sigma_{\beta}(V)=\widehat{\Gamma}_{\beta}\cdot V \quad \text{for all}\hspace{2mm} V\in\mathbb{R}^{6},
\end{equation}
where $\widehat{\Gamma}_{\beta}\in\mathbb{R}^{6}$ is the unit vector
\begin{equation}\label{gammavector}
\widehat{\Gamma}_{\beta}:=\frac{1}{\sqrt{m^{2}d_{\beta}^{2}+2J^{2}}}\left[
\begin{array}{c}
0 \\
md_{\beta}e(\psi)^{\perp} \\
J \\
J
\end{array}
\right].
\end{equation}
We shall term a canonical scattering map $\sigma_{\beta}$ which satisfies the algebraic identities above a {\bf physical} canonical scattering map. We begin with the following elementary observation.
\begin{prop}
Suppose $\{T_{t}\}_{t\in\mathbb{R}}$ is a physical regular flow. Its associated canonical scattering family $\{\sigma_{\beta}\}_{\beta\in\mathbb{T}^{3}}$ has the following properties:
\begin{enumerate}[{\em (S1)}]
\item $\sigma_{\beta}:\mathcal{V}_{\beta}^{-}\rightarrow\mathcal{V}_{\beta}^{+}$ is a $C^{1}$ diffeomorphism;
\item $\sigma_{\beta}$ is a classical solution of either the Jacobian PDE
\begin{equation}\label{negone}
\mathrm{det}(D\sigma_{\beta}(V))=-1
\end{equation}
or 
\begin{equation}\label{posone}
 \mathrm{det}(D\sigma_{\beta}(V))=1
\end{equation}
for $V\in\mathrm{int}\,\mathcal{V}_{\beta}^{-}$;
\item $\sigma_{\beta}$ satisfies the conservation laws \eqref{scatcolm}, \eqref{scatcoam} and \eqref{scatcoke}.
\end{enumerate} 
\end{prop}
\begin{proof}
This follows simply from the definition of physical regular flow on $\mathscr{D}_{2}(\mathsf{P}_{\ast})$.
\end{proof}
The above proposition states that any canonical scattering map $\sigma_{\beta}$ is a classical solution of a Jacobian partial differential equation on the open half-space $\mathcal{V}_{\beta}^{-}$. Note that knowing $\{T_{t}\}_{t\in\mathbb{R}}$ to be only a physical regular flow is not enough information for us to deduce boundary conditions for these PDE on the boundary hyperplane 
\begin{equation*}
\{V\in\mathbb{R}^{6}\,:\,V\cdot M\widehat{\nu}_{\beta}=0\}.
\end{equation*}
It is precisely this degree of freedom in the scattering problem that gives rise to non-uniqueness of physical regular flows associated to {\em linear} scattering maps in section \ref{maspy} below.
\subsection{\label{sec:level2}Derivation of Physical Scattering Maps}
In this section, before proceding to the set-up and analysis of so-called Monge-Amp\`{e}re scattering problems, we follow the approach of the work of \textsc{Fr\'{e}mond} \cite{fremond} and \textsc{Palffy-Muhoray, Virga, Wilkinson and Zheng} \cite{palffy2017paradox} in the derivation of physical scattering maps for compact, strictly convex sets with boundary curves of class $\mathscr{C}^{1}$. In particular, for a given collision configuration $\beta\in\mathbb{T}^{3}$, we adopt the {\em impulse ansatz} that allows one to close the system of algebraic equations \eqref{scatcolm}, \eqref{scatcoam} and \eqref{scatcoke} for the unknown post-collisional velocities $v_{\beta}', \ov{v}_{\beta}', \omega_{\beta}'$ and $\ov{\omega}_{\beta}'$ in $\mathcal{V}_{\beta}^{+}$, given pre-collisional velocities $v, \ov{v}, \omega$ and $\ov{\omega}$ in $\mathcal{V}_{\beta}^{-}$. In this article, we show that this oft-employed ansatz is sufficient, but not necessary, to guarantee the non-interpenetration of two compact, strictly convex sets.
\subsubsection{The Algebraic Physical Scattering Problem for $\mathscr{P}_{2}(\mathsf{P}_{\ast})$}
Suppose $\beta\in\mathbb{T}^{3}$ and $V\in\mathcal{V}^{-}_{\beta}$ have been given. By writing the unknown velocity components in a vector $\xi\in\mathbb{R}^{6}$, it is straightforward to see that \eqref{colm}, \eqref{coam} and \eqref{coke} are equivalent to the 4 algebraic equations
\begin{align}
\mathrm{(A)}\quad \left\{
\begin{array}{c}
|M\xi|^{2}=\rho^{2}\vspace{2mm}\notag\\
c_{i}\cdot \xi=d_{i} \quad \text{for}\hspace{2mm}i=1, 2, 3,
\end{array}
\right. \notag
\end{align}
where $\rho, c_{i}$ and $d_{i}$ are given by
\begin{align}
\rho:=|MV|^{2}, \vspace{2mm}\notag \\
c_{1}:=\widehat{E}_{1} \qquad d_{1}:=\widehat{E}_{1}\cdot V, \vspace{2mm}\notag \\
c_{2}:=\widehat{E}_{2} \qquad d_{2}:=\widehat{E}_{2}\cdot V, \vspace{2mm}\notag \\
c_{3}:=\widehat{\Gamma}_{\beta} \qquad d_{1}:=\widehat{\Gamma}_{\beta}\cdot V.\notag
\end{align}
This problem appears to be underdetermined in the sense that, ostensibly, there are too few algebraic equations\footnote{We have found, whilst discussing this problem with collaborators, that counting equations is indeed a poor approach to the derivation of scattering maps. Indeed, with a sphere in $\mathbb{R}^{M}$ and one of its tangent planes in mind, it is possible to study $M$ linear equations and 1 nonlinear equation which admit a unique solution in $\mathbb{R}^{M}$.} to solve uniquely for the unknown $\xi\in\mathbb{R}^{6}$. Indeed, in the proof of theorem below we show there exists an {\em uncountable} family of solutions of (A) indexed by elements of the real projective line $\mathbb{RP}^{1}$. For the moment, however, we assume there exists an {\em impulse parameter} $\alpha=\alpha(V; \beta)\in\mathbb{R}$ such that
\begin{equation}\label{alphone}
m\left(
\begin{array}{c}
\xi_{1} \\
\xi_{2}
\end{array}
\right)=mv-\alpha n_{\beta} \quad \text{and} \quad m\left(
\begin{array}{c}
\xi_{3} \\
\xi_{4}
\end{array}
\right)=m\ov{v}+\alpha n_{\beta} \tag{IA1}
\end{equation}
together with
\begin{equation}\label{alphtwo}
J\xi_{5}=J\omega-\alpha p^{\perp}_{\beta}\cdot n_{\beta} \quad \text{and}\quad J\xi_{6}=J\omega+\alpha(p_{\beta}-d_{\beta}e(\psi))^{\perp}\cdot n_{\beta}. \tag{IA2}
\end{equation}
Substitution of \eqref{alphone} and \eqref{alphtwo} into the algebraic identities yields the following quadratic in the unknown impulse parameter $\alpha$:
\begin{equation}\label{quadratic}
Q_{2}\alpha^{2}+Q_{1}\alpha=0,
\end{equation}
where 
\begin{equation*}
Q_{2}(V; \beta):=\Lambda_{\beta}
\end{equation*}
and
\begin{equation*}
Q_{1}(V; \beta):=-2(v+\omega p_{\beta}^{\perp}-\ov{v}-\ov{\omega}q_{\beta}^{\perp})\cdot n_{\beta}.
\end{equation*}
One solution of \eqref{quadratic} is clearly given by $\alpha(V; \beta):=0$ for all $V\in\mathbb{R}^{6}$ and $\beta\in\mathbb{T}^{3}$. The other solution is given by
\begin{equation*}
\alpha(V; \beta):=2\Lambda_{\beta}^{-1}\left(v+\omega p_{\beta}^{\perp}-\ov{v}-\ov{\omega}q_{\beta}^{\perp}\right)\cdot n_{\beta},
\end{equation*}
which yields that $s_{\beta}:\Sigma_{\beta}^{-}\rightarrow\Sigma_{\beta}^{+}$ is a {\em linear} map on $\mathbb{R}^{6}$ given by
\begin{equation}\label{wee}
s_{\beta}:=M^{-1}(I-2\widehat{\nu}_{\beta}\otimes\widehat{\nu}_{\beta})M,
\end{equation}
recalling that $\widehat{\nu}_{\beta}\in\mathbb{R}^{6}$ is the unit vector
\begin{equation*}
\widehat{\nu}_{\beta}:=\frac{1}{\sqrt{\frac{2}{m}+\frac{1}{J}|p_{\beta}^{\perp}\cdot n_{\beta}|^{2}+\frac{1}{J}|q_{\beta}^{\perp}\cdot n_{\beta}|^{2}}}M^{-1}\left[
\begin{array}{c}
-n_{\beta} \\ 
n_{\beta} \\ 
-p_{\beta}^{\perp}\cdot n_{\beta} \\ 
q_{\beta}^{\perp}\cdot n_{\beta}
\end{array}
\right].
\end{equation*}
We note that, under the impulse ans\"{a}tze \eqref{alphone} and \eqref{alphtwo}, the map $\sigma_{\beta}$ is not only a {\em linear} map on $\mathbb{R}^{6}$, but is conjugate by the mass-inertia matrix $M$ to a reflection in $\mathrm{O}(6)$. It is curious to note that no delineation between pre- and post-collisional velocity vectors in $\mathbb{R}^{6}$ was made in the derivation of the above matrix.

We claim that \eqref{wee} is not the only family of canonical scattering maps which (i) conserves linear momentum, angular momentum and kinetic energy, and (ii) has the property that $\beta\mapsto s_{\beta}V$ is a smooth map on $\mathbb{T}^{3}$ for $V\in\mathbb{R}^{6}$. We now introduce the so-called Monge-Amp\`{e}re scattering problem which allows us to substantiate this claim.
\section{The Monge-Amp\`{e}re Scattering Problem}\label{maspy}
In this section, we adopt what we believe to be a new approach to the derivation of physical canonical scattering families $\{s_{\beta}\}_{\beta\in\mathbb{T}^{3}}$. Rather than introduce ans\"{a}tze on the form of the canonical scattering map $s_{\beta}$, we focus our attention on the Jacobian PDE \eqref{negone} and \eqref{posone}, one of which members of any scattering family associated to a physical regular flow on $\mathscr{P}_{2}(\mathsf{P}_{\ast})$ necessarily satisfies. Evidently, one must make a choice as to whether the regular flow has the property $\mathrm{det}(D\sigma_{\beta})=-1$ or $\mathrm{det}(D\sigma_{\beta})=-1$ on $\mathbb{R}^{6}$, i.e. if it ought to be orientation-preserving or -reversing in velocity space. 
\subsection{Statement of the Physical Scattering Problem}
As we have seen above, the properties of physical regular flows on $\mathscr{D}_{2}(\mathsf{P}_{\ast})$ gives rise to the following scattering problem:
\subsubsection*{The Physical 2-Body Scattering Problem}
Let $\beta\in\mathbb{T}^{3}$ be given. Find a $C^{1}$ diffeomorphism $\sigma_{\beta}:\Sigma_{\beta}^{-}\rightarrow\Sigma_{\beta}^{+}$ which is a classical solution of the first-order PDE
\begin{equation}\label{jaccypde}
\mathrm{det}\,(D\sigma_{\beta}(V))=\pm 1\quad \text{on}\hspace{2mm}\mathrm{int}\,\Sigma_{\beta}^{-}
\end{equation}
subject to the following constraints:
\begin{enumerate}[(SP1)]
\item $\sigma_{\beta}^{2}(V)=V$ for all $V\in\Sigma_{\beta}^{-}$;\vspace{1.5mm}
\item $\widehat{E}_{1}\cdot \sigma_{\beta}(V)=\widehat{E}_{1}\cdot V$ and $\widehat{E}_{2}\cdot \sigma_{\beta}(V)=\widehat{E}_{2}\cdot V$; \vspace{2mm}
\item $\widehat{\Gamma}_{\beta}\cdot \sigma_{\beta}(V)=\widehat{\Gamma}_{\beta}\cdot V$; \vspace{2mm}
\item $|M\sigma_{\beta}(V)|^{2}=|MV|^{2}$.
\end{enumerate}
Evidently, the physical 2-body scattering problem asks one to characterise all self-inverse $C^{1}$ diffeomorphisms between $\Sigma_{\beta}^{-}$ and $\Sigma_{\beta}^{+}$ which respect the conservation laws of classical physics. We do not investigate the full problem as it is stated here, as it is enough for the proof of our main result (theorem \ref{mainrez}) to consider classical solutions of \eqref{jaccypde} of {\em potential} form. Indeed, this leads us naturally to the study of a class of Monge-Amp\`{e}re equations on $\mathbb{R}^{6}$.
\subsection{Statement of the Monge-Amp\`{e}re Physical Scattering Problem.}
Let us assume that for each $\beta\in\mathbb{T}^{3}$, the scattering map $\sigma_{\beta}$ satisfying \eqref{jaccypde} is of potential form, i.e. $\sigma_{\beta}:=\nabla S_{\beta}$ for some $C^{2}$ scalar function $S_{\beta}:\Sigma_{\beta}^{-}\rightarrow \mathbb{R}$. The corresponding orientation-preserving or orientation-reversing {\em Monge-Amp\`{e}re physical scattering problem} (or simply MASP$\pm$, respectively) for $S_{\beta}$ can be stated as follows:
\subsubsection*{The 2-Body MASP$\pm$}
Let $\beta\in\mathbb{T}^{3}$ be given. Find a $C^{2}$ map $S_{\beta}:\Sigma_{\beta}^{-}\rightarrow \mathbb{R}$ satisfying
\begin{equation}\label{mongey}
\mathrm{det}\,D^{2}S_{\beta}=\pm1\quad \text{on}\hspace{2mm}\mathrm{int}\,\Sigma_{\beta}^{-}
\end{equation}
subject to the conditions
\begin{enumerate}[(\textrm{MA}1)]
\item $\nabla S_{\beta}(\nabla S_{\beta}(V))=V$ for all $V\in\Sigma_{\beta}^{-}$; \vspace{1.5mm}
\item $\widehat{E}_{1}\cdot \nabla S_{\beta}(V)=\widehat{E}_{1}\cdot V$ and $\widehat{E}_{2}\cdot \nabla S_{\beta}(V)=\widehat{E}_{2}\cdot V$; \vspace{2mm}
\item $\widehat{\Gamma}_{\beta}\cdot \nabla S_{\beta}(V)=\widehat{\Gamma}_{\beta}\cdot V$; \vspace{2mm}
\item $|M\nabla S_{\beta}(V)|^{2}=|MV|^{2}$.
\end{enumerate}
We claim that there are {\em at least} 2 solutions of the MASP- problem and {\em uncountably-many} solutions of the MASP+ problem (parameterised by elements of the real projective line $\mathbb{RP}^{1}$). In turn, each of these physical scattering families gives rise to a distinct physical regular flow on $\mathscr{P}_{2}(\mathsf{P}_{\ast})$ by the work in \cite{ball}. 
\begin{rem}
Our approach herein is guided by the sequence of works of \textsc{Pogorelov} \cite{pogorelov1972improper}, \textsc{Calabi} \cite{calabi1958improper}, \textsc{J\"{o}rgens} \cite{jorgens1954losungen} and \textsc{Cheng and Yau} \cite{cheng1986complete} on entire solutions of the Monge-Am\`{e}re equation. Indeed, using the result of their work we are able to characterise all entire canonical physical scattering maps $\sigma_{\beta}$ on $\Sigma_{\beta}^{-}$. We are unable to obtain an analogous characterisation result for those scattering maps which are of class $C^{k}$, $C^{\infty}$ or real analytic on $\Sigma_{\beta}^{-}$
\end{rem}
Let us now state and prove the above claims.
\begin{thm}\label{maps1}
The MASP$-$ problem admits 2 distinct {\em quadratic} solutions of the form 
\begin{equation}\label{mapsansatz}
S_{\beta}(V)=\frac{1}{2}V\cdot (M^{-1}A_{\beta}M)V
\end{equation}
for $A_{\beta}\in\mathrm{O}(6)$, namely
\begin{equation*}
A_{\beta}^{(1)}:=I-2\widehat{\nu}_{\beta}\otimes\widehat{\nu}_{\beta}
\end{equation*}
and
\begin{equation}\label{secondscat}
A_{\beta}^{(2)}:=2\widehat{E}_{1}\otimes\widehat{E}_{1}+2\widehat{E}_{2}\otimes\widehat{E}_{2}+2\widehat{E}_{\beta}\otimes\widehat{E}_{\beta}-I,
\end{equation}
where $\widehat{E}_{1}$, $\widehat{E}_{2}$ and $\widehat{\nu}_{\beta}$ are defined by \eqref{eeee} and \eqref{eetah} above, respectively, and 
\begin{equation}\label{gsgamma}
\widehat{E}_{\beta}:= \frac{1}{\sqrt{2md_{\beta}^{2}+8J}}\left(
\begin{array}{c}
\sqrt{m}d_{\beta}\sin\psi \\
-\sqrt{m} d_{\beta}\cos\psi \\
-\sqrt{m} d_{\beta}\sin\psi \\
\sqrt{m} d_{\beta}\cos\psi \\
2\sqrt{J}\\
2\sqrt{J}
\end{array}
\right).
\end{equation}
In particular, $\mathcal{S}_{1}$ and $\mathcal{S}_{2}$ defined by
\begin{equation*}
\mathcal{S}_{j}:=\left\{s^{(j)}_{\beta}\right\}_{\beta\in\mathbb{T}^{3}} \quad \text{with}\hspace{2mm}s_{\beta}^{(1)}(V):=M^{-1}A_{\beta}^{(j)}MV
\end{equation*}
for $j=1, 2$ are two distinct families of canonical physical scattering maps on $\mathbb{R}^{6}$.
\end{thm}
\begin{proof}
It proves helpful to transform the Monge-Amp\`{e}re scattering problem in order to reveal the r\^{o}le of orthogonal matrices in linear scattering. Indeed, suppose that $S_{\beta}(V)$ is of the form \eqref{mapsansatz} above. It follows that the matrix $A_{\beta}\in\mathbb{R}^{6\times 6}$ satisfies
\begin{equation}\label{detty}
\mathrm{det}\,A_{\beta}=-1.
\end{equation}
Moreover, we find from (MA2) that
\begin{equation*}
\widehat{E}_{1}\cdot (M^{-1}A_{\beta}M)V=\widehat{E}_{1}\cdot V \quad \Longleftrightarrow\quad (A_{\beta}^{T}M^{-1}\widehat{E}_{1}-M^{-1}\widehat{E}_{1})\cdot MV=0
\end{equation*}
for all $V\in\mathbb{R}^{6}$, from which it follows by definition of $M$ that $\widehat{E}_{1}$ is a unit eigenvector of $A_{\beta}^{T}$ with corresponding eigenvalue 1. Similarly, from (MA2), (MA3) and the Gram-Schmidt algorithm one has that both $\widehat{E}_{2}$ and $\widehat{E}_{\beta}$ (defined in \eqref{gsgamma} above) are unit eigenvectors of $A_{\beta}^{T}$ each with associated eigenvalue 1. Moreover, $\{\widehat{E}_{1}, \widehat{E}_{2}, \widehat{E}_{\beta}\}$ constitutes a set of mutually-orthogonal vectors in $\mathbb{R}^{6}$.

By definition of scattering map, as $M^{-1}A_{\beta}M$ must be a bijection between $\Sigma_{\beta}^{-}$ and $\Sigma_{\beta}^{+}$, it follows that $A_{\beta}$ must be a bijection between $\widehat{\Sigma}_{\beta}^{-}$ and $\widehat{\Sigma}_{\beta}^{+}$, where
\begin{equation*}
\widehat{\Sigma}_{\beta}^{-}:=\left\{W\in\mathbb{R}^{6}\,:\,W\cdot\widehat{\nu}_{\beta}\leq 0\right\}\quad \text{and}\quad\widehat{\Sigma}_{\beta}^{-}:=\left\{W\in\mathbb{R}^{6}\,:\,W\cdot\widehat{\nu}_{\beta}\geq 0\right\}.
\end{equation*}
This is only possible if $A_{\beta}\widehat{\nu}_{\beta}=-\widehat{\nu}_{\beta}$, which yields that $\widehat{\nu}_{\beta}$ is a unit eigenvector of $A_{\beta}^{T}$ with associated eigenvalue $-1$. One can check that $\widehat{\nu}_{\beta}$ is orthogonal to $\mathrm{span}\,\{\widehat{E}_{1}, \widehat{E}_{2}, \widehat{E}_{\beta}\}$.

Next, (MA4) holds if and only if $|A_{\beta}W|^{2}=|W|^{2}$ for all $W\in\mathbb{R}^{6}$, whence $A_{\beta}\in\mathrm{O}(6)$. Moreover, (MA1) reduces to the identity $A_{\beta}^{2}=I$ in $\mathrm{O}(6)$, and so $A_{\beta}$ is self-adjoint. By the Spectral Decomposition Theorem (see \textsc{Bollob\'{a}s} \cite{bollobas1990linear}, p. 200), it follows that $A_{\beta}$ admits the representation formula
\begin{equation}\label{reppy}
A_{\beta}=\widehat{E}_{1}\otimes \widehat{E}_{1}+\widehat{E}_{2}\otimes \widehat{E}_{2}+ \widehat{E}_{\beta}\otimes\widehat{E}_{\beta}-\widehat{\nu}_{\beta}\otimes\widehat{\nu}_{\beta}+\lambda_{1}\widehat{F}_{1}\otimes \widehat{F}_{1}+\lambda_{2}\widehat{F}_{2}\otimes \widehat{F}_{2},
\end{equation}
for some mutually-orthogonal unit vectors $\widehat{F}_{1}=\widehat{F}_{1}(\beta)$ and $\widehat{F}_{2}=\widehat{F}_{2}(\beta)$ with the property that
\begin{equation}\label{basis}
\left\{\widehat{E}_{1}, \widehat{E}_{2}, \widehat{E}_{\beta}, \widehat{\nu}_{\beta}, \widehat{F}_{1}, \widehat{F}_{2}\right\}\subset\mathbb{R}^{6}
\end{equation}
is a basis for $\mathbb{R}^{6}$; moreover, $\lambda_{1}, \lambda_{2}\in\mathbb{R}$ are eigenvalues associated to $\widehat{F}_{1}$ and $\widehat{F}_{2}$, respectively. Noting that  $\widehat{E}_{1}, \widehat{E}_{2}, \widehat{E}_{\beta}$ and $\widehat{E}_{\beta}$ are mutually orthogonal, it follows that $\widehat{F}_{1}$ and $\widehat{F}_{2}$ both lie in the subspace $\mathrm{span}\{\widehat{E}_{1}, \widehat{E}_{2}, \widehat{E}_{\beta}, \widehat{\nu}_{\beta}\}^{\perp}$.

By the representation formula \eqref{reppy}, it follows from \eqref{detty} that
\begin{equation*}
\lambda_{1}\lambda_{2}=1.
\end{equation*}
As $A_{\beta}$ is idempotent on $\mathbb{R}^{6}$, it follows that $\lambda_{1}^{2}=\lambda_{2}^{2}=1$, and so it is either the case that $\lambda_{1}=\lambda_{2}=-1$ or $\lambda_{1}=\lambda_{2}=1$. Finally, using the fact that
\begin{align}\label{woop}
\widehat{F}_{1}\otimes\widehat{F}_{1}+\widehat{F}_{2}\otimes\widehat{F}_{2} \vspace{2mm}\notag \\
= I-\widehat{E}_{1}\otimes\widehat{E}_{1}-\widehat{E}_{2}\otimes\widehat{E}_{2}-\widehat{E}_{\beta}\otimes\widehat{E}_{\beta}-\widehat{\nu}_{\beta}\otimes\widehat{\nu}_{\beta}
\end{align}
in $\mathbb{R}^{6\times 6}$, the result follows.
\end{proof}
\begin{rem}
We note that the matrix \eqref{secondscat} corresponds to the following expressions for post-collisional velocities, given $V\in\Sigma_{\beta}^{-}$:
\begin{align}
v_{\beta}':=\ov{v}+\frac{d_{\beta}}{md_{\beta}^{2}+4J}\left(md_{\beta}v-2J\omega e(\psi)^{\perp}-md_{\beta}\ov{v}-2J\ov{\omega}e(\psi)^{\perp}\right)\cdot e(\psi)^{\perp}e(\psi)^{\perp}, \notag \vspace{2mm}\\
\ov{v}_{\beta}':=v-\frac{d_{\beta}}{md_{\beta}^{2}+4J}\left(md_{\beta}v-2J\omega e(\psi)^{\perp}-md_{\beta}\ov{v}-2J\ov{\omega}e(\psi)^{\perp}\right)\cdot e(\psi)^{\perp}e(\psi)^{\perp}, \notag \vspace{2mm} \\
\omega_{\beta}':=-\omega-\frac{2}{md_{\beta}^{2}+4J}\left(md_{\beta}v-2J\omega e(\psi)^{\perp}-md_{\beta}\ov{v}-2J\ov{\omega}e(\psi)^{\perp}\right)\cdot e(\psi)^{\perp},\notag \vspace{2mm}\\
\ov{\omega}_{\beta}':=-\ov{\omega}-\frac{2}{md_{\beta}^{2}+4J}\left(md_{\beta}v-2J\omega e(\psi)^{\perp}-md_{\beta}\ov{v}-2J\ov{\omega}e(\psi)^{\perp}\right)\cdot e(\psi)^{\perp}.
\end{align}
It is clear from this formulation that impulse is proportional to the vector $e(\psi)^{\perp}$ and {\em not} the normal to the point of collision $n_{\beta}$. One might then legitimately ask which of these two collision boundary conditions is `physically' appropriate.
\end{rem}
We now consider the `wilder' case of orientation-preserving scattering.
\begin{thm}\label{maps2}
The MASP+ problem admits uncountably-many quadratic solutions of the form 
\begin{equation}\label{mapsansatz2}
S_{\beta}(V)=\frac{1}{2}V\cdot (M^{-1}A_{\beta}M)V
\end{equation}
for $A_{\beta}\in\mathrm{O}(6)$ of the form
\begin{equation*}
A_{\beta}:=I-2\widehat{\nu}_{\beta}\otimes\widehat{\nu}_{\beta}-2\widehat{F}_{\beta}\otimes\widehat{F}_{\beta},
\end{equation*}
where $\widehat{F}_{\beta}$ is any unit vector in $\mathrm{span}\{\widehat{E}_{1}, \widehat{E}_{2}, \widehat{E}_{\beta}, \widehat{\nu}_{\beta}\}$. As such, if one has that $\beta\mapsto\widehat{F}_{\beta}$ lies in $C^{0}(\mathbb{T}^{3}, \mathbb{R}^{6})$, then $\{\nabla S_{\beta}\}_{\beta\in\mathbb{T}^{3}}$ with $S_{\beta}$ defined by \eqref{mapsansatz2} is a canonical physical scattering family.
\end{thm}
\begin{proof}
This follows the same lines as the proof of theorem \ref{maps1}, with the difference that $\lambda_{1}\lambda_{2}=-1$. For the purpose of utilising the Spectral Decomposition Theorem, one can pick any two orthogonal vectors which span $\{\widehat{E}_{1}, \widehat{E}_{2}, \widehat{E}_{\beta}, \widehat{\nu}_{\beta}\}^{\perp}$; equivalently, one can choose any unit vector $\widehat{F}_{\beta}$ in $\{\widehat{E}_{1}, \widehat{E}_{2}, \widehat{E}_{\beta}, \widehat{\nu}_{\beta}\}^{\perp}$ and generate the last remaining unit eigendirection $\widehat{G}_{\beta}$ of $A_{\beta}$ by means of the Gram-Schmidt process. However, using the fact that
\begin{equation*}
\widehat{G}_{\beta}\otimes \widehat{G}_{\beta}=I-\widehat{E}_{1}\otimes\widehat{E}_{1}-\widehat{E}_{2}\otimes\widehat{E}_{2}-\widehat{E}_{\beta}\otimes\widehat{E}_{\beta}-\widehat{\nu}_{\beta}\otimes\widehat{\nu}_{\beta}-\widehat{F}_{\beta}\otimes \widehat{F}_{\beta}
\end{equation*}
owing to the difference in the signs of the eigenvalues, one now has that \eqref{reppy} becomes
\begin{equation*}
A_{\beta}=I-2\widehat{\nu}\otimes\widehat{\nu}-2\widehat{F}_{\beta}\otimes\widehat{F}_{\beta},
\end{equation*}
and so we are done.
\end{proof}
This result shows that one can associate to every continuous map $\phi:\mathbb{T}^{2}\rightarrow\mathbb{RP}^{1}$ a family of canonical physical scattering maps, each of which is orientation-preserving on $\mathbb{R}^{6}$. Using Theorems \ref{maps1} and \ref{maps2} above, it follows from the general existence theory of \textsc{Ballard} \cite{ball} that there exist uncountably-many physical regular flows on $\mathscr{D}_{2}(\mathsf{P}_{\ast})$. Indeed, for the convenience of the reader, we recast the relevant result of \cite{ball} in the terminology of this work.
\begin{thm}\label{preciseball}
Suppose that $\mathsf{P}_{\ast}\subset\mathbb{R}^{2}$ is a compact, strictly-convex set with real-analytic boundary curve $\partial\mathsf{P}_{\ast}$. Let $\mathcal{S}:=\{\sigma_{\beta}\}_{\beta\in\mathbb{T}^{3}}$ be an associated family of physical canonical scattering maps. For each $Z_{0}\in\mathscr{D}_{2}(\mathsf{P}_{\ast})$ with $X_{0}\in\mathrm{int}\,\mathscr{P}_{2}(\mathsf{P}_{\ast})$, there exists a unique global-in-time physical weak solution $t\mapsto Z(t)$ of \eqref{newty}. This solution depends explicitly on the choice of $\mathcal{S}$. Moreover, $t\mapsto Z(t)$ is either (i) real analytic on $\mathbb{R}$, or (ii) real analytic on $\mathbb{R}$ outside the countable set of isolated collision times $\mathcal{T}(Z_{0})$.
\end{thm}
\begin{proof}
See \textsc{Ballard} \cite{ball}, theorem 9 and corollary 9.
\end{proof}
Finally, the proof of our main claim follows quickly. \vspace{2mm}

{\em Proof of Theorem \ref{mainrez}}. The set of all $Z_{0}$ in $\mathrm{int}\,\mathscr{P}_{2}(\mathsf{P}_{\ast})\times \mathbb{R}^{6}$ such that 
\begin{equation*}
\left\{
X_{0}+tV_{0}\,:\,t\in\mathbb{R}
\right\}
\end{equation*}
has non-empty intersection with $\partial\mathscr{P}_{2}(\mathsf{P}_{\ast})$ is of non-null $(\mathscr{L}_{6}\mres\mathscr{P}_{2}(\mathsf{P}_{2}))\otimes \mathbb{R}^{6}$-measure. For any $Z_{0}$ in this set, we apply Theorem \ref{preciseball} above using any choice of canonical scattering family one wishes.
\section{Discussion: Relevance of Result for the Boltzmann Equation}\label{discuss}
We claim that our non-uniqueness result has interesting consequences for the Boltzmann equation, and for kinetic theory in general. As we shall discuss below, it prompts -- among others -- the following natural questions:
\begin{enumerate}[(Q1)]
\item Do all those collision operators built using the distinct families of canonical physical scattering maps derived in this work admit the same `macroscopic properties'? For instance, do they all admit the same set of collision invariants?
\item Is there a `macroscopic' difference between linear scattering maps and nonlinear scattering maps at the level of particle dynamics?
\end{enumerate}
In any case, in this final section, we shall fix our attention on the Boltzmann equation which governs a probability distribution function $f=f(x, \vartheta, v, \omega, t)$ that models the `average' behaviour of a gas comprised of congruent strictly convex hard particles which do not admit the symmetry group of the disk. If the IBVP for \eqref{nemo} modelling the underlying $N$-particle system is furnished with boundary conditions\footnote{To be more precise, the component of the boundary $\partial\mathcal{P}_{N}(\mathsf{P}_{\ast})$ corresponding to {\em 2-body collisions} is furnished with the boundary conditions effected by $\{\sigma_{\beta}\}_{\beta\in\mathbb{T}^{3}}$.} effected by a canonical physical scattering family $\{\sigma_{\beta}\}_{\beta\in\mathbb{T}^{3}}$, it can be shown by appealing to the formal derivation via the BBGKY hierarchy (see \textsc{Gallagher, Saint-Raymond and Texier} \cite{gallagher2013newton} in the case of hard disks, for instance) the equation reads as
\begin{equation}\label{bolty}
\frac{\partial f}{\partial t}+(v\cdot\nabla_{x})f+\omega\frac{\partial f}{\partial\vartheta}=\mathcal{C}[f, f],
\end{equation}
where $\mathcal{C}$ is the collision operator given by
\begin{equation}\label{collster}
\mathcal{C}[f, f]:=\int_{\mathbb{R}^{2}\times\mathbb{R}\times\mathbb{S}^{1}}\int_{\mathbb{S}^{1}}|V\cdot \widehat{\nu}_{\beta}|\left(f_{\beta}'\ov{f}_{\beta}'-f\ov{f}\right)\, dS(n)dS(\ov{\vartheta})d\ov{\omega}d\ov{v},
\end{equation}
with
\begin{equation*}
\begin{array}{c}
f_{\beta}'=f(x, \vartheta, v_{\beta}', \omega_{\beta}', t), \qquad \ov{f}_{\beta}'=f(x, \ov{\vartheta}, \ov{v}_{\beta}', \ov{\omega}_{\beta}', t), \\
f=f(x, \vartheta, v, \omega, t), \qquad \ov{f}=f(x, \ov{\vartheta}, \ov{v}, \ov{\omega}, t),
\end{array}
\end{equation*}
and
\begin{equation*}
\left[
\begin{array}{c}
v_{\beta}' \\
\ov{v}_{\beta}'\\
\omega_{\beta}' \\
\ov{\omega}_{\beta}'
\end{array}
\right]:=\sigma_{\beta}[V]
\end{equation*}
for $x\in\mathbb{R}^{2}$, $t\geq 0$, $V=[v, \ov{v}, \omega, \ov{\omega}]\in\mathbb{R}^{6}$. 

In order to fix ideas, we proceed at a formal level, i.e. we do not worry about establishing any notion of solution to the above kinetic equation. Suppose one is interested in characterising all spatially-homogeneous equilibrium solutions $M$ of \eqref{bolty}. In other words, one is interested in finding all {\em Maxwellian distributions} satisfying the equation
\begin{equation}\label{maxy}
\mathcal{C}[M, M]=0.
\end{equation}
To do this, one multiplies throughout equation \eqref{maxy} by $\log{M}$, integrates over phase space $\mathbb{R}^{2}\times\mathbb{R}\times\mathbb{S}^{1}$, and uses the fact that the change of variables $V\mapsto \sigma_{\beta}[V]$ has unit Jacobian on $\mathbb{R}^{6}$ to produce the identity
\begin{equation}\label{dude}
\int_{\mathbb{R}^{6}\times\mathbb{T}^{3}}|V\cdot\widehat{\nu}_{\beta}|\left(M_{\beta}'\ov{M}_{\beta}-M\ov{M}\right)\log\left(\frac{M_{\beta}'\ov{M}_{\beta}'}{M\ov{M}}\right)\,dVd\beta=0.
\end{equation}
By elementary properties of the natural logarithm, it follows that identity \eqref{dude} holds if and only if
\begin{equation}\label{papapa}
\begin{array}{c}
M_{\beta}'\ov{M}_{\beta}'=M\ov{M} \vspace{1mm}\\
\Longleftrightarrow\quad 
\log{M_{\beta}'}+\log{\ov{M}_{\beta}'}=\log{M}+\log{\ov{M}} \vspace{1mm}\\
\Longleftrightarrow \quad \phi(v_{\beta}', \omega_{\beta}', \vartheta)+\phi(\ov{v}_{\beta}', \ov{v}_{\beta}', \ov{\vartheta})=\phi(v, \omega', \vartheta)+\phi(\ov{v}', \ov{v}', \ov{\vartheta}),
\end{array}
\end{equation}
for all $(v, \omega, \vartheta)\in\mathbb{R}^{2}\times\mathbb{R}\times\mathbb{S}^{1}$, where $\phi(v, \omega, \vartheta):=\log{M(v, \omega, \vartheta)}$. As such, $M$ is a Maxwellian if and only if $\log{M}$ is a {\bf collision invariant} (see \textsc{Cercignani, Illner and Pulvirenti} \cite{cercignani2013mathematical}, Chapter 3, for more on such topics in the case of spherical particles).

As mentioned in the introduction of this article, this problem has already been studied in \cite{lsrmw} in the case of the canonical physical scattering family containing the maps
\begin{equation}\label{berp}
\sigma_{\beta}:=M^{-1}\left(I-2\widehat{\nu}_{\beta}\otimes \widehat{\nu}_{\beta}\right)M, \quad \text{for}\hspace{2mm}\beta\in\mathbb{T}^{3}.
\end{equation}
In the process of establishing this result, one rewrites identity \eqref{papapa} as
\begin{equation*}
\Phi(\sigma_{\beta}[V]; \vartheta, \ov{\vartheta})=\Phi(V; \vartheta, \ov{\vartheta}) \quad \text{for all}\hspace{2mm}V\in\mathbb{R}^{6}, \vspace{2mm}\beta\in\mathbb{T}^{3},
\end{equation*}
where $\Phi:\mathbb{R}^{6}\times\mathbb{T}^{2}\rightarrow\mathbb{R}$ is defined pointwise as $\Phi(V; \vartheta, \ov{\vartheta}):=\phi(v, \omega, \vartheta)+\phi(\ov{v}, \ov{\omega}, \ov{\vartheta})$. As such, for a fixed spatial parameter $(\vartheta, \ov{\vartheta})\in\mathbb{T}^{2}$, one looks to characterise all scalar invariants of the group $G_{(\vartheta, \ov{\vartheta})}\subseteq\mathrm{O}(6)$ generated algebraically by the set
\begin{equation*}
\left\{
\sigma_{\beta}\,:\,\psi\in\mathbb{S}^{1}
\right\}.
\end{equation*}
In \cite{lsrmw}, the following result was proved.
\begin{thm}\label{olestuff}
Suppose $\mathsf{P}_{\ast}\subset\mathbb{R}^{2}$ is a compact, strictly-convex set with real-analytic boundary curve $\partial\mathsf{P}_{\ast}$. Moreover, suppose $\mathsf{P}_{\ast}$ admits the following symmetries:
\begin{equation*}
(I-2e\otimes e)\mathsf{P}_{\ast}=\mathsf{P}_{\ast} \quad \text{and}\quad (I-2e^{\perp}\otimes e^{\perp})\mathsf{P}_{\ast}=\mathsf{P}_{\ast} 
\end{equation*}
for some unit vector $e\in\mathbb{R}^{2}$. Suppose a measurable map $\phi:\mathbb{R}^{2}\times\mathbb{R}\times\mathbb{S}^{1}\rightarrow\mathbb{R}$ satisfies the identity
\begin{equation}
\phi(v_{\beta}', \omega_{\beta}', \vartheta)+\phi(\ov{v}_{\beta}', \ov{v}_{\beta}', \ov{\vartheta})=\phi(v, \omega', \vartheta)+\phi(\ov{v}', \ov{v}', \ov{\vartheta})
\end{equation} 
for all $V\in\mathbb{R}^{6}$ and $\beta=(\vartheta, \ov{\vartheta}, \psi)\in\mathbb{T}^{3}$, where the post-collisional variable $V_{\beta}':=[v_{\beta}', \ov{v}_{\beta}', \omega_{\beta}', \ov{\omega}_{\beta}']$ is determined by the scattering matrix \eqref{berp}. Then $\phi$ is necessarily of the form
\begin{equation}
\phi(v, \omega, \vartheta)=a(\vartheta)+b\cdot v+c(m|v|^{2}+J\omega^{2})
\end{equation}
for some measurable function $a:\mathbb{S}^{1}\rightarrow\mathbb{R}$, $b\in\mathbb{R}^{2}$ and $c\in\mathbb{R}$.
\end{thm}
As a quick consequence of this result, it follows that when one chooses to build the collision operator \eqref{collster} using the canonical scattering family $\{\sigma_{\beta}\}_{\beta\in\mathbb{T}^{3}}$ with $\sigma_{\beta}$ as in \eqref{berp} above, all Maxwellia of the associated Boltzmann equation are of the shape
\begin{equation}
M(v, \omega, \vartheta):=\frac{1}{|\mathbb{S}^{1}|}\frac{m}{2\pi \Theta}\sqrt{\frac{J}{2\pi \Theta}}e^{-\frac{m|v-u|^{2}+J\omega^{2}}{\Theta}},
\end{equation}
for some constant $u\in\mathbb{R}^{2}$ and $\Theta>0$. However, the method of proof in \cite{lsrmw} makes crucial use of the fact that each scattering matrix \eqref{berp} is conjugate, by the mass-inertia matrix $M$, to a reflection in $\mathrm{O}(6)$. As such, for every other family of canonical physical scattering matrices derived in section \ref{maspy} above, the analogue of Theorem \ref{olestuff} is unknown.

Why ought one care about this observation that there are uncountably-many physical scattering matrices? From a mathematical point of view, this problem is connected to the characterisation of classical solutions of the Monge-Amp\`{e}re equation
\begin{equation}
\mathrm{det}\,D^{2}\sigma=\mathrm{const.} \quad \text{on}\hspace{2mm}\mathbb{R}^{M}
\end{equation}
for $M\geq 2$, when the map $\sigma$ is specified to be (i) of class $\mathscr{C}^{k}$, (ii) of class $\mathscr{C}^{\infty}$, or (iii) of class $\mathscr{C}^{\omega}$. To the knowledge of the author, this has remained an open problem since the work of \textsc{Cheng and Yau} \cite{cheng1986complete}. This problem is also connected to the (perhaps, somewhat surprising) absence of a general existence theory for the IBVP for \eqref{nemo} when $\partial\mathsf{P}_{\ast}$ is only of class $\mathscr{C}^{\infty}$ (as opposed to $\mathscr{C}^{\omega}$). Moreover, collision invariants are directly linked with relaxation to equilibrium in kinetic equations (see \textsc{Desvillettes and Villani} \cite{desvillettes2005trend}) and also to hydrodynamic limits thereof (see \textsc{Saint-Raymond} \cite{saint2009hydrodynamic}). From a physical point of view, it is important to know whether or not the Boltzmann equation is a {\em universal} equation -- in the class of all compact, strictly-convex sets -- for the average dynamics of rarefied particle systems. Given its connection to questions in both mathematics and physics, we believe the observation herein warrants attention in the future.
\subsubsection*{Acknowledgements}
I would like to extend my sincere thanks and gratitude to Gilles Francfort and to Patrick Ballard for several extended conversations and encouragement during the preparation of this manuscript.

\appendix

\bibliography{euler}

\end{document}